\documentclass[sigconf]{acmart}
\renewcommand\footnotetextcopyrightpermission[1]{}
\setcopyright{none} 

\usepackage{xspace}
\usepackage{amsfonts}
\usepackage{cleveref}
\usepackage{graphicx}
\usepackage{times}

\usepackage{tikz}
\usepackage{amsmath}

\usepackage[utf8]{inputenc}
\usepackage[normalem]{ulem}
\usepackage[font=it,labelfont=bf]{caption}
 \usepackage{subcaption}
 \usepackage{enumerate}

\usepackage{algorithmicx}
\usepackage{algorithm}
\usepackage{algpseudocode}
\usepackage{paralist}

\algnewcommand\Not{\textbf{ not }}
\algnewcommand\AndT{\textbf{ and }}
\algnewcommand\OrT{\textbf{ or }}
\algnewcommand\In{\textbf{ in }}

\algsetblockdefx[ForAll]{ForAll}{EndForAll}{}{}[2]{\textbf{for} $#1 \in #2$ \textbf{do}}{\textbf{end for}}
\algsetblockdefx[Parameters]{Parameters}{EndParameters}{}{}{\textbf{Parameters:}}{}
\algsetblockdefx[Struct]{Struct}{EndStruct}{}{}[1]{\textbf{Struct} $#1$ \textbf{contains}}{}
\algsetblockdefx[Function]{Function}{EndFunction}{}{}[2]{\textbf{function} \emph{#1}($#2$) \textbf{:}}{}
\algsetblockdefx[Process]{Process}{EndProcess}{}{}[2]{\textbf{process} \emph{#1}($#2$) \textbf{:}}{}
\algsetblockdefx[Upon]{Upon}{EndUpon}{}{}[1]{\textbf{upon} $#1$ \textbf{do}}{}
\algsetblockdefx[UponExists]{UponExists}{EndUponExists}{}{}[2]{\textbf{upon exists} $#1$ \textbf{such that} $#2$ \textbf{do}}{}
\algsetblockdefx[UponCondition]{UponCondition}{EndUponCondition}{}{}[1]{\textbf{upon} $#1$ \textbf{do}}{}
\algsetblockdefx[UponReceiving]{UponReceiving}{EndUponReceiving}{}{}[1]{\textbf{upon receiving} $#1$ \textbf{do}}{}
\algsetblockdefx[UponEvent]{UponEvent}{EndUponEvent}{}{}[1]{\textbf{upon event} $#1$ \textbf{do}}{}
\algsetblockdefx[Implements]{Implements}{EndImplements}{}{}{\textbf{Implements:}}{}
\algsetblockdefx[Uses]{Uses}{EndUses}{}{}{\textbf{Uses:}}{}

\newcommand{\sysname}{ISS\xspace}
\newcommand{\sysnameFull}{Insanely Scalable SMR\xspace}
\newcommand{\modulename}{Sequenced Broadcast\xspace}
\newcommand{\moduleAbbr}{SB\xspace}
\newcommand{\moduleAbbrSmall}{sb\xspace}
\newcommand{\moduleProperty}{SB}
\newcommand{\smr}{SMR\xspace}
\newcommand{\smrProperty}{SMR}
\newcommand{\tob}{TOB\xspace}

\newcommand{\maxSnE}[1]{max(Sn(#1))}
\newcommand{\segment}[2]{Seg(#1,#2)}
\newcommand{\maxSnS}[2]{max(\segment{#1}{#2})}
\newcommand{\epLeaders}[1]{Leaders(#1)}
\newcommand{\numLeaders}[1]{|\epLeaders{#1}|}
\newcommand{\initBuckets}[2]{initBuckets(#1, #2)}
\newcommand{\extraBuckets}[1]{extraBuckets(#1)}
\newcommand{\activeBuckets}[2]{Buckets(#1, #2)}

\newcommand{\bcast}{\textsc{SMR-CAST}}
\newcommand{\deliver}{\textsc{SMR-DELIVER}}

\newcommand{\suspect}{\textsc{SUSPECT}}
\newcommand{\restore}{\textsc{RESTORE}}

\newcommand{\rbcast}{\textsc{BRB-CAST}}
\newcommand{\rbdeliver}{\textsc{BRB-DELIVER}}
\newcommand{\rbProperty}{BRB}

\newcommand{\bcpropose}{\textsc{BC-PROPOSE}}
\newcommand{\bcdecide}{\textsc{BC-DECIDE}}
\newcommand{\bcProperty}{BC}

\newcommand{\sbcast}{\textsc{SB-CAST}}
\newcommand{\sbinit}{\textsc{SB-INIT}}
\newcommand{\sbdeliver}{\textsc{SB-DELIVER}}

\newcommand{\checkpoint}{\textsc{CHECKPOINT}\xspace}
\newcommand{\stableCheckpoint}{\textsc{STABLE-CHECKPOINT}\xspace}
\newcommand{\request}{\textsc{REQUEST}\xspace}

\begin{document}
\fancyhf{} 
\fancyhead{}
\fancyfoot[C]{\thepage}

\date{}

\title{State-Machine Replication Scalability Made Simple\\ (Extended Version)}

\author{Chrysoula Stathakopoulou}
\affiliation{%
  \institution{IBM Research Europe - Zurich}
  \country{}
}
\author{Matej Pavlovic}
\affiliation{%
  \institution{Protocol Labs}
  \country{}
}
\authornote{Work done while at IBM Research - Zurich}

\author{Marko Vukoli\'c}
\affiliation{%
  \institution{Protocol Labs}
  \country{}
}
\authornote{Work done while at IBM Research - Zurich}

\begin{abstract}
Consensus, state-machine replication (SMR) and total order broadcast (TOB) protocols
are notorious for being poorly scalable with the number of participating nodes.
Despite the recent race to reduce overall message complexity of leader-driven SMR/TOB protocols,
scalability remains poor and the throughput is typically inversely proportional to the number of nodes.
We present Insanely Scalable State-Machine Replication, a generic construction to turn leader-driven protocols into scalable multi-leader ones.
For our scalable SMR construction we use a novel primitive
called Sequenced (Total Order) Broadcast (SB) which we wrap around PBFT, HotStuff and Raft leader-driven protocols to make them scale.
Our construction is general enough to accommodate most leader-driven ordering protocols (BFT or CFT) and make them scale.
Our implementation
improves the peak throughput of
PBFT, HotStuff, and Raft by 37x, 56x, and 55x, respectively, at a scale of 128 nodes.
\end{abstract}

\maketitle

\section{Introduction}

Considerable research effort has recently been dedicated to scaling state-machine replication (SMR) and total-order broadcast (TOB) or consensus protocols,
fundamental primitives in distributed computing.
By scaling, we mean maintaining high throughput and low latency despite a growing number of nodes (replicas) $n$.
Driven by the needs of blockchain systems, particular focus lies on Byzantine fault-tolerant (BFT) protocols
in the eventually synchronous (deterministic protocols) or fully asynchronous (randomized protocols) model.

In this model, the classical Dolev/Reischuk (DR) \cite{DR-85} lower bound requires $\Omega(n^2)$ worst case message complexity,
which was a focal complexity metric of many subsequent protocols including recent HotStuff~\cite{hotstuff}.
However, we claim message complexity to be a rather poor scalability metric,
demonstrated by the fact that HotStuff and other leader-driven protocols scale
inversely proportionally to the number of nodes, despite some of them matching the DR lower bound.
This is because in leader-driven protocols, the leader has at least $O(n)$ bits to send, even in the common case,
yielding $n^{-1}$ throughput scalability.

A recent effort to overcome the single leader bottleneck
by allowing multiple parallel leaders (Mir-BFT~\cite{stathakopoulou2019mir}) in the classical PBFT protocol \cite{Castro:2002:PBF}
demonstrates high scalability in practice.
Despite  certain advantages of PBFT,
e.g., being highly parallelizable and designed not to require signatures on protocol messages,
among the many (existing and future) TOB solutions, there are none that fit all use cases.
HotStuff~\cite{hotstuff} is the first protocol with linear message complexity both in common case and in leader replacement, making it suitable for highly asynchronous or faulty networks.
On the other hand, other protocols, e.g., Aliph/Chain~\cite{700paper}, have optimal throughput when failures are not expected to occur often.
Finally, crash fault-tolerant (CFT) protocols such as Raft~\cite{RAFT} and Paxos~\cite{lamport2001paxos} tolerate a larger number of (benign) failures than BFT protocols for the same number of nodes.

Our work takes the Mir-BFT effort one step further,
introducing by introducing \sysnameFull, hereinafter referred to as \sysname, the first modular framework to make leader-driven \tob protocols scale.
\sysname is a classic SMR system that establishes a total order of client requests with typical liveness and safety properties,
applicable to any replicated service, such as resilient databases or a blockchain ordering layer (e.g., as in Hyperledger Fabric~\cite{AndroulakiBBCCC18}).

Notably, and unlike previous efforts~\cite{stathakopoulou2019mir}\cite{avarikioti2020fnf}, \sysname achieves scalability without requiring a primary node to periodically decide on the protocol configuration.
\sysname achieves this by introducing a novel abstraction, \modulename (\moduleAbbr), which requires each instance of an ordering protocol to \emph{terminate} after delivering a finite number of messages.
This allows nodes in \sysname to decide on the configuration independently and deterministically, without requiring additional communication and without relying on a single primary node.

This in turn allows for more flexible and fair leader selection policies.
Moreover, it guarantees better resilience against an adaptive adversary that can corrupt the primary node, which changes slowly and can be known in advance with a deterministic round robin rotation~\cite{stathakopoulou2019mir}\cite{avarikioti2020fnf}.

\sysname implements SMR by multiplexing multiple instances of \moduleAbbr which
operate concurrently on a \emph{partition} of the domain of client requests.
We carefully select the partition to maintain safety and liveness, as well as to prevent redundant data duplication, which has been shown to be detrimental to performance~\cite{stathakopoulou2019mir}.
This is qualitatively better than related modular efforts~\cite{gupta2019scaling,BFT-Mencius}
which do not provide careful partitioning and load balancing and hence cannot achieve the same scalability and robustness at the same time.

\sysname maintains a \emph{contiguous} log of (batches of) client requests at each node.
Each position in the log corresponds to a unique sequence number
and \sysname agrees on the assignment of a unique request batch to each sequence number.
Our goal is to introduce as much parallelism as possible in assigning batches to sequence numbers
while avoiding \emph{request duplication}, i.e., assigning the same request to more than one sequence number.
To this end, \sysname subdivides the log into non-overlapping \emph{segments}.
Each segment, representing a subset of the log's sequence numbers, corresponds to an independent instance of \moduleAbbr
that has its own leader and executes concurrently with other \moduleAbbr instances.
\moduleAbbr abstraction, moreover, facilitates the reasoning about multiplexing the outputs of multiple instances into a single log,
while staying very close to classic definitions of broadcast and thus being easily implementable by existing algorithms.

To prevent the leaders of two different segments from concurrently proposing the same request,
and thus wasting resources, while also preventing malicious leaders from censoring
(i.e., not proposing) certain requests,
we adopt and generalize the rotating bucketization of the request space introduced by MirBFT~\cite{stathakopoulou2019mir}.
\sysname assigns a different \emph{bucket}---subset of client requests---to each segment.
No bucket is assigned to more than one segment at a time and each request maps (through a hash function) to exactly one bucket.
\sysname periodically changes the bucket assignment, such that each bucket is guaranteed to eventually be assigned to a segment with a correct leader.

To maintain the invariant of one bucket being assigned to one segment, all buckets need to be re-assigned at the same time.
\sysname therefore uses finite segments that it groups into \emph{epochs}.
An epoch is a union of multiple segments that forms a contiguous sub-sequence of the log.
After all log positions within an epoch have been assigned request batches,
and thus no requests are ``in-flight'',
\sysname advances to the next epoch, meaning that it starts processing a new set of segments
forming the next portion of the log.

We implement and deploy \sysname on a wide area network (WAN) spanning 16 different locations spread around the world,
demonstrating \sysname{}s performance using two different BFT protocols (PBFT\cite{Castro:2002:PBF} and Hotstuff\cite{hotstuff})
and one CFT protocol (Raft\cite{RAFT}).
On 128 nodes \sysname improves the performance of the single leader counterpart protocols, PBFT, HotStuff, and Raft, by 37x, 56x, and 55x, respectively.

The rest of this paper is organized as follows.
\Cref{sec:foundation} presents the theoretical foundation of our work.
It models the systems we study (\Cref{sec:model}),
introduces the \moduleAbbr abstraction (\Cref{sec:sb}) and
describes how we multiplex \moduleAbbr instances with \sysname (\Cref{sec:multiplexing,sec:req-dupl}).
\Cref{sec:details,sec:implementation} respectively describe the details of \sysname and its implementation,
including three different leader-driven protocols.
In \Cref{sec:correctness} we prove that \moduleAbbr can be implemented with consensus and that multiplexing \moduleAbbr instances with
\sysname implements \smrProperty.
In \Cref{sec:evaluation} we evaluate the performance of \sysname.
In \Cref{sec:related} we discuss related work.
We conclude in \Cref{sec:conclusion}.

\section{Theoretical Foundations}
\label{sec:foundation}

 This section describes in detail the principles and abstractions underlying \sysname.
We first define the system model, and present our core \modulename (\moduleAbbr) abstraction.
We then introduce \sysname which uses \moduleAbbr instances to implement state machine replication.

\subsection{System Model}
\label{sec:model}

We assume a set $\mathcal{N}$ of \emph{node} processes with $|\mathcal{N}|=n$.
At most, $f$ of the nodes in $\mathcal{N}$ can fail.
We further assume a set $\mathcal{C}$ of \emph{client} processes of arbitrary size, any of which can be faulty.
Each process is identified by its public key, provided by a public key infrastructure.
Unless mentioned otherwise, we assume Byzantine, i.e., arbitrary, faults.
Therefore, we require $n\ge 3f+1$.
We further assume that nodes in $\mathcal{N}$ are computationally bounded and cannot subvert standard cryptographic primitives.

Processes communicate through authenticated point-to-point channels.
We assume a partially synchronous network~\cite{Dwork:1988:CPP:42282.42283}
such that the communication between any pair of correct processes is asynchronous before an unknown time $GST$,
when the communication becomes synchronous.

Nodes in $\mathcal{N}$ implement a state machine replication (\smr) service to clients in $\mathcal{C}$.
To broadcast request $r$, a client $c$ triggers an $\langle \bcast | r\rangle$ event.
A client request is a tuple $r=(o, id)$, where $o$ is the request payload, e.g., some operation to be executed by some application, and $id$ a unique request identifier.
The request identifier is a tuple $id=(t,c)$ where $t$ is a logical timestamp and $c$ a client identifier, e.g., a client public key.
Two client requests $r=(o,id),r'=(o',id')$ are considered equal,
we write $r=r'$ and we refer to them as duplicates, if and only if
$o=o'\land id=id'$.

Nodes assign a unique sequence number $sn$ to $r$ and eventually output an $\langle\deliver | sn, r\rangle$ event such that the following properties hold:

\noindent \textbf{\smrProperty1 {Integrity}: }
If a correct node delivers $(sn, r)$, where $r.id.c$ is a correct client's identidy,
then client $c$ broadcast $r$.\\
\noindent \textbf{\smrProperty2 Agreement:}
If two correct nodes deliver, respectively, $(sn,r)$ and $(sn,r')$, then $r=r'$.\\
\noindent \textbf{\smrProperty3 Totality:}
If a correct node delivers $(sn,r)$, then every correct node eventually delivers $(sn,r)$.\\
\noindent \textbf{\smrProperty4 Liveness:}
If a correct client broadcasts request $r$, then some correct node eventually delivers $(sn,r)$.

 \paragraph{No-duplication.}
 Our \smr implementation with \sysname further guarantees that if a correct node delivers request $(r,sn)$ and  $(r,sn')$, then  $sn = sn'$.
 Notice that invoking of \smr with distinct sequence numbers allows to prevent the execution of duplicate requests: an application can trivially filter out $r,sn'$ when $(r,sn)$ with $sn'>sn$ is already executed.
 However, enforcing  that no correct node invokes \smr with a duplicate request is critical for performance.

\subsection{\modulename (\moduleAbbr)}
\label{sec:sb}

\modulename (\moduleAbbr) is a variant of Byzantine total order broadcast~\cite{cachin2010introduction}
with explicit sequence numbers and an explicit set of allowed messages.

\moduleAbbr is instantiated with a failure detector instance as a parameter.
We assume an \emph{eventually strong failure detector} in an environment with Byzantine faults denoted as  $\Diamond S(bz)$, as defined by Malkhi and Reiter~\cite{malkhi1997unreliable}.

In particular, we assume that each node has access to a local failure detector module $D$ which provides a list of suspected nodes $D.suspected$.
When a node $p$ is present in the list of suspects of a node $q$ we say that $q$ \emph{suspects} $p$.
Each node's list of suspects may change over time and may differ from the lists of other nodes.
We denote with $\langle\suspect | p \rangle$  the event of adding a node $p$ to the list suspects.
We denote with $\langle\restore | q \rangle$  the event of removing a node $q$ from the list of suspects.

The failure detector D of the class $\Diamond S(bz)$ detects \emph{quiet} nodes.
Intuitively, a quiet node is the equivalent to a crashed node in the BFT model,
accounting for non-crash faults that are indistinguishable from crashes.
For the exact definition we refer the reader to Malkhi and Reiter's work~\cite{malkhi1997unreliable}.

A failure detector of the $\Diamond S(bz)$ class guarantees the following two properties:

\noindent\textbf{Strong Completeness:} There is a time after which every quiet node is permanently suspected by every correct node.\\
\noindent\textbf{Eventual Weak Accuracy:} There is a time after which some correct node is never suspected by any correct node.

Notice that a faulty node is not suspected unless quiet, even in if it sends malformed messages.
In other words, the failure detector is only concerned with the lack of messages, not with their content.

We can now define \modulename as follows.
Let $M$ be a set of messages and $S \subseteq \mathbb{N}$ a set of sequence numbers.
Only one designated sender node $\sigma\in \mathcal{N}$ can broadcast messages (we hereon write sb-cast to distinguish from other primitives) by invoking $\langle \sbcast | sn, m \rangle$ with $(sn, m) \in S \times M$.\\
$\langle \sbdeliver | sn, m \rangle$ is triggered at a correct node $p$ when $p$ delivers (we hereon write sb-delivers) message $m$ with sequence number $sn$.

If a correct node suspects that $\sigma$ is quiet, all correct nodes are allowed to sb-deliver a special \emph{nil} value $m = \bot \notin M$.
If, however, $\sigma$ is trusted by all correct nodes,
all correct nodes are guaranteed to sb-deliver non-nil messages $m \neq \bot$.

\moduleAbbr is explicitly initialized with an $\langle \sbinit \rangle$ event.
We assume a failure detector list at each correct node which is initially empty.
It is only after the invocation of $\langle \sbinit \rangle$ that suspecting $\sigma$ can lead to the $\bot$ value being delivered.

\noindent An instance of $\moduleAbbr(\sigma, S, M, D)$ has the following properties:

\noindent\textbf{\moduleProperty1 Integrity:}
If a correct node sb-delivers $(sn, m)$
with $m\neq\bot$ and $\sigma$ is correct then $\sigma$ sb-cast $(sn, m)$.\\
\noindent \textbf{\moduleProperty2 Agreement:}
If two correct nodes sb-deliver, respectively, $(sn, m)$ and $(sn, m')$, then $m = m'$.\\
\noindent\textbf{\moduleProperty3 Termination:}
If $p$ is correct, then $p$ eventually sb-delivers a message for every sequence number in $S$, i.e., $\forall sn \in S : \exists m \in M \cup \{\bot\}$
such that $p$ sb-delivers $(sn, m)$.\\
\noindent\textbf{\moduleProperty4 Eventual Progress:}
If some correct node sb-delivers $(sn,\bot)$ for some $sn\in S$, then some correct node $p$ suspected $\sigma$ after $sb$ is initialized at $p$.

The key differences in comparison to \tob are that:
\begin{itemize}
  \item \moduleAbbr is invoked for an explicit set of sequence numbers $S$ and messages $M$.
  \item \moduleAbbr is invoked with a $\Diamond S(bz)$ failure detector.
  \item correct nodes deliver messages from set $M$ \emph{and} the special $\bot$ value.
  \item \moduleAbbr terminates for all seqence numbers.
\end{itemize}
The latter is guaranteed by $\bot$ value and $\Diamond S(bz)$ completeness;
if $\sigma$ is quiet it will eventually be suspected by all correct nodes.

\subsection{Multiplexing Instances of \moduleAbbr with \sysname}
\label{sec:multiplexing}
\sysname multiplexes instances of \moduleAbbr to implement \smr.
Each node maintains a log of ordered messages which correspond to batches of client requests.
Each position in the log corresponds to a \emph{sequence number} signifying the offset from the start of the log.
The log is partitioned in subsets of sequence numbers called \emph{segments}. Each segment corresponds to one instance of \moduleAbbr.
Nodes obtain requests from clients and, after mapping them to a log position using an instance of \moduleAbbr,
deliver them together with the assigned sequence number.

\sysname proceeds in \emph{epochs} identified by monotonically increasing integer \emph{epoch numbers}.
Each epoch $e$ is associated with a set of segments. The union of those segments forms a set $Sn(e)$ of consecutive sequence numbers.
Epoch $0$ (the first epoch) starts with sequence number 0.
The mapping of sequence numbers to epochs is a function known to all nodes with the only requirements being that it is monotonically increasing and that there are no gaps between epochs.
More formally, $max(Sn(e)) + 1 = min(Sn(e + 1))$.
Epoch length can be arbitrary, as long as it is finite.
For simplicity, we use a fixed, constant epoch length.

\begin{figure}[h]
  \centering
  \includegraphics[width=\columnwidth]{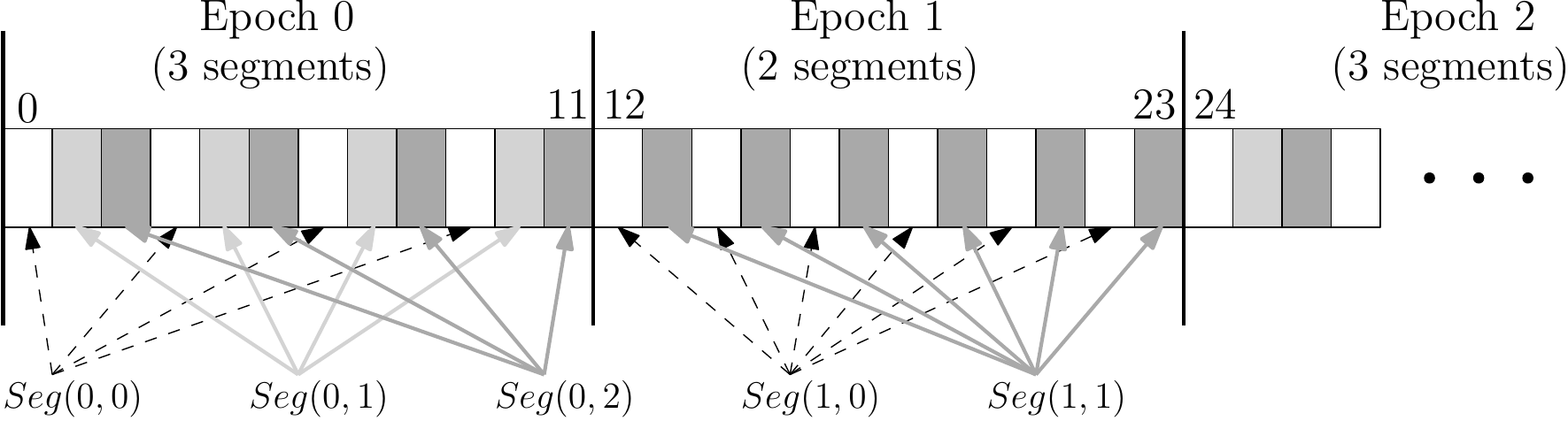}
  \caption{\label{fig:epochs-segments} Log partitioned in epochs and segments.
  In this particular example, each epoch is 12 sequence numbers long.
  The first epoch has 3 segments while the second epoch only has 2,
  i.e.,  $\numLeaders{0} = 3$, $\numLeaders{1} = 2$, $\maxSnE{1} = 23$ and $\maxSnS{0}{1} = 10$.}
\end{figure}

Epochs are processed sequentially, i.e., \sysname first agrees on the assignment of request batches to all sequence numbers in $Sn(e)$
before starting to agree on the assignment of request batches to sequence numbers in $Sn(e+1)$.

Within an epoch, however, \sysname processes segments in parallel.
Multiple leaders, selected according to a \emph{leader selection policy},
concurrently propose batches of requests for different sequence numbers in $Sn(e)$.
To this end, \sysname assigns a different \emph{leader} node to each segment in epoch $e$.
We refer to the set of all nodes acting as leaders in an epoch as the leaderset of the epoch.
The numbers of leaders and segments in each epoch always match.
We denote by $\segment{e}{i}$ the subset of $Sn(e)$ for which node $i$ is the leader.
This means that node $i$ is responsible for proposing request batches to sequence numbers in $\segment{e}{i}$.
No node other than $i$ can propose batches for sequence numbers in $\segment{e}{i}$.
Let $\epLeaders{e}$ be nodes that are leaders in epoch $e$.
We associate sequence numbers with segments in a round-robin way, namely, for $0 \leq i < \numLeaders{e}$,
\begin{equation*}
\segment{e}{i} \subseteq Sn(e) = \{sn \in Sn(e) | i \equiv sn \textrm{ mod } \numLeaders{e} \}
\end{equation*}
An example with $\numLeaders{0} = 3$ and $\numLeaders{1} = 2$ is illustrated in \Cref{fig:epochs-segments}.

In principal, any assignment of sequence numbers to segments is possible and leads to a correct algorithm.
We choose the round-robin assignment because it uniformly distributes sequence numbers among instances.
Therefore, in a fault-free execution it is the least likely to create ``gaps'' in the log, which minimizes the end-to-end request latency.

In order not to waste resources on duplicate requests,
we require that a request cannot be part of two batches assigned to two different sequence numbers.
We enforce this at three levels: (1) within a segment, (2) across segments in the same epoch, and (3) across epochs.

Within a segment, we rely on the fact that
the correct leader will propose (and a correct node, as follower, will accept)
only batches with disjoint sets of requests for each sequence number within a segment.
Across segments,
we partition the set of all possible requests into \emph{buckets} using a hash function
and enforce that only requests from different buckets
can be used for different segments within an epoch.
We denote by $\mathcal{B}$ the set of all possible buckets.
We assign a subset of $\mathcal{B}$ to each segment, such that each bucket is assigned to exactly one segment in each epoch.
We denote by $\activeBuckets{e}{i} \subseteq \mathcal{B}$ the set of buckets assigned to leader $i$ in epoch $e$.%
\footnote{We sloppily say that we assign a bucket to a leader $i$
when assigning a bucket to a segment for which $i$ is the leader.}
Across epochs, \sysname  prevents duplication by only allowing a node to propose a request batch in a new epoch once it has added all batches from the previous epoch to  the log.
If a request has been delivered in a batch in the previous epoch, a correct leader will not propose it again (see also \Cref{sec:epochs}).
Also, a correct node, as follower, will not accept a proposal which includes a previously delivered request.

In summary, a segment of epoch $e$ with leader $i$ is defined by the tuple $(e, i, \segment{e}{i}, \activeBuckets{e}{i})$.

For a set of buckets $B\subseteq\mathcal{B}$ , we denote with $batches(B)$ the set of all possible batches consisting of valid (we define request validity precisely later in \Cref{sec:requests}) requests that map to some bucket in $B$.
For each segment $(e,\allowbreak i,\allowbreak \segment{e}{i},\allowbreak \activeBuckets{e}{i})$,
we use an instance $\moduleAbbr(i,\allowbreak batches(\activeBuckets{e}{i}),\allowbreak \segment{e}{i}, \allowbreak D)$ of \modulename, where $D$ is the above-mentioned failure detector.
We say that leader $i$ \emph{proposes} a batch $b\allowbreak \in \allowbreak batches(\activeBuckets{e}{i})$ \allowbreak for sequence  number $sn \allowbreak\in \allowbreak \segment{e}{i}$
if $i$ executes $\sbcast(sn, b)$ at the corresponding instance of \moduleAbbr.
A batch $b$ \emph{commits} with sequence number $sn$ (and is added to the log at the corresponding position) at node $n$ when
the corresponding instance of \moduleAbbr triggers $\sbdeliver(sn, b)$ at node $n$.

During epoch $e$, all nodes that are leaders in $e$
simultaneously propose batches for sequence numbers in their corresponding segments.
\sysname multiplexes all segments into the single common log as shown in \Cref{fig:epochs-segments}.
Each node thus executes $\numLeaders{e}$ \moduleAbbr instances simultaneously,
while being a leader for at most one of them.

Epoch $e$ ends and epoch $(e+1)$ starts when all sequence numbers in $Sn(e)$ have been committed.
Nodes keep the old instances active
until all corresponding sequence numbers become part of a stable checkpoint.
This is necessary for ensuring totality (even for slow nodes which might not have actively taken part in the agreement).

\subsection{Assigning Buckets to Segments}
\label{sec:req-dupl}

\sysname partitions the request hash space into buckets
which it assigns to leaders/segments and changes this bucket assignment at epoch transitions.
At any point in time, a leader can assign sequence numbers only to requests from its assigned buckets.
This approach was first used in Mir-BFT \cite{stathakopoulou2019mir} to counter request duplication and censoring attacks.

During an epoch, the assignment of buckets to leaders is fixed.
To ensure liveness, each bucket must repeatedly be assigned to a correct leader.
To this end, \sysname re-assigns the buckets on every epoch transition as follows (illustrated in \Cref{fig:buckets}).

\begin{figure}[h]
  \centering
  \includegraphics[width=\columnwidth]{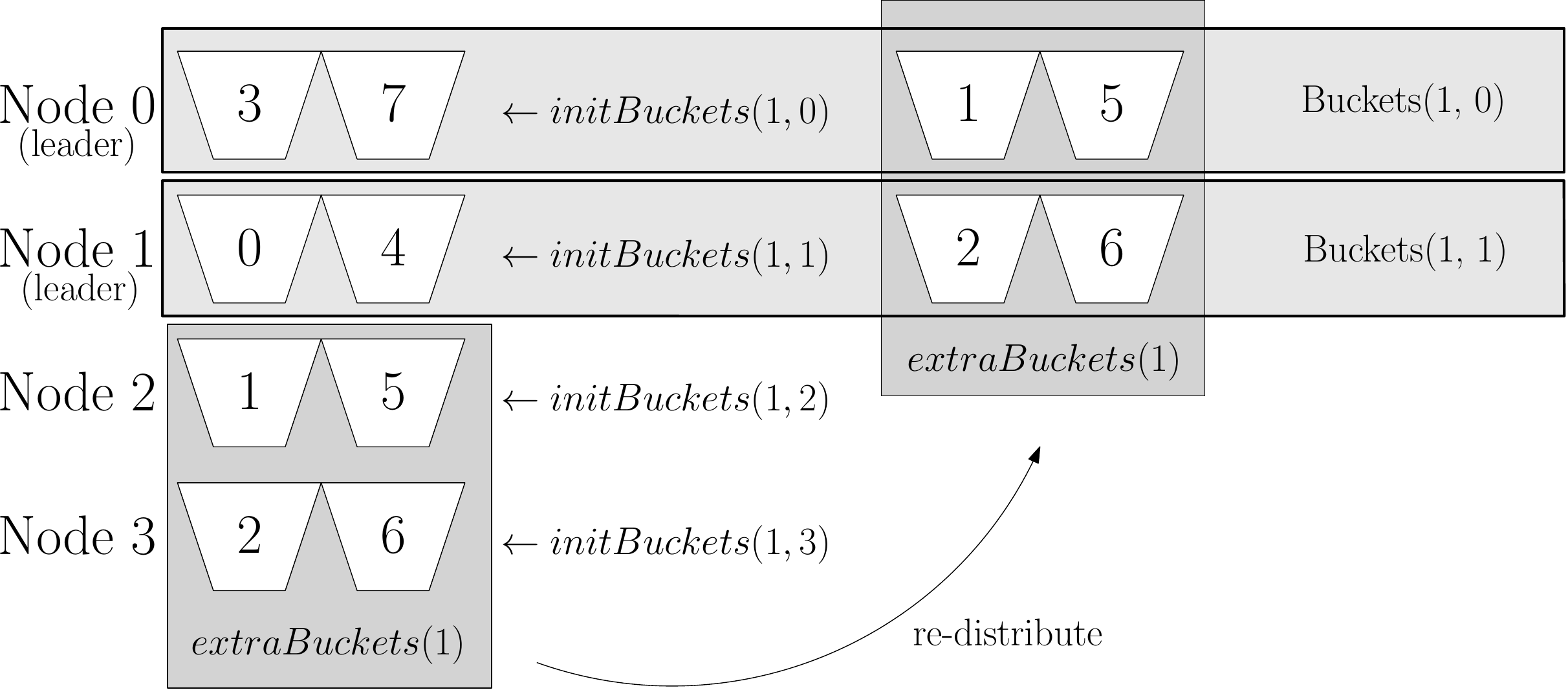}
  \caption{\label{fig:buckets} Example assignment of 8 buckets to 2 leaders in a system with 4 nodes in epoch 1.}
\end{figure}

For epoch $e$, we start by assigning an initial set of buckets to \emph{each node} (leader or not) in a round-robin way.
Let $\initBuckets{e}{i} \subseteq \mathcal{B}$ be the set of buckets initially assigned to each node $i$, $0 \leq i < n$ in epoch $e$.
We consider the buckets in $\mathcal{B}$ to be numbered, with each bucket having an integer bucket number $b \in \{0,\dots,|\mathcal{B})|-1\}$.
In the following we refer to buckets using $b$.
\begin{equation}
  \initBuckets{e}{i} = \{b \in \mathcal{B} \hspace{.4em} | \hspace{.4em} (b + e) \equiv i \textrm{ mod } n\}
  \label{eq:buckets}
\end{equation}
However, not all nodes belong to $\epLeaders{e}$.
Let $\extraBuckets{e}$ be the set of buckets initially assigned to non-leaders.

\begin{align*}
  \extraBuckets{e} = \{b \in \mathcal{B} \hspace{.4em} | \hspace{.4em} \exists i : & \hspace{.4em} i \notin \epLeaders{e} \\
                                                                                   & \hspace{.4em} \wedge b \in \initBuckets{e}{i} \}
\end{align*}

We must re-distribute those extra buckets to the leaders of epoch $e$.
We do this in a round robin way as well.
Let $l(e, k)$, $0 \leq k < \numLeaders{e}$ be the $k$-th leader (in lexicographic order) in epoch $e$.
The $\activeBuckets{e}{l(e, k)}$ of the $k$-th leader in $e$ are thus defined as follows.
\begin{align*}
  \activeBuckets{e}{l(e, k)} & = & & \initBuckets{e}{l(e, k)} \text{\ } \cup\\
  & & & \{b \in \extraBuckets{e} \hspace{.4em} |\\
  & & & (b + e) \equiv k \textrm{ mod } \numLeaders{e}\}
\end{align*}
An example bucket assignment is illustrated in \Cref{fig:buckets}.

With this approach, all buckets are assigned to leaders
and every node is eventually assigned every bucket at least through the initial bucket assignment.
\sysname ensures liveness as long as, in an infinite execution,
there is a correct node that (1) eventually stops being suspected forever by every correct node,
and (2) is assigned each bucket infinitely often.
(1) is satisfied by the properties of the eventually strong failure detector.
(2) is satisfied by the bucket re-assignment and the leader selection policies described in \Cref{sec:leader-selection}.

\Cref{fig:epochs-segments} shows how \sysname multiplexes \moduleAbbr instances to build a totally ordered log.

\begin{figure}[h!]
\centering
\includegraphics[width=0.95\columnwidth]{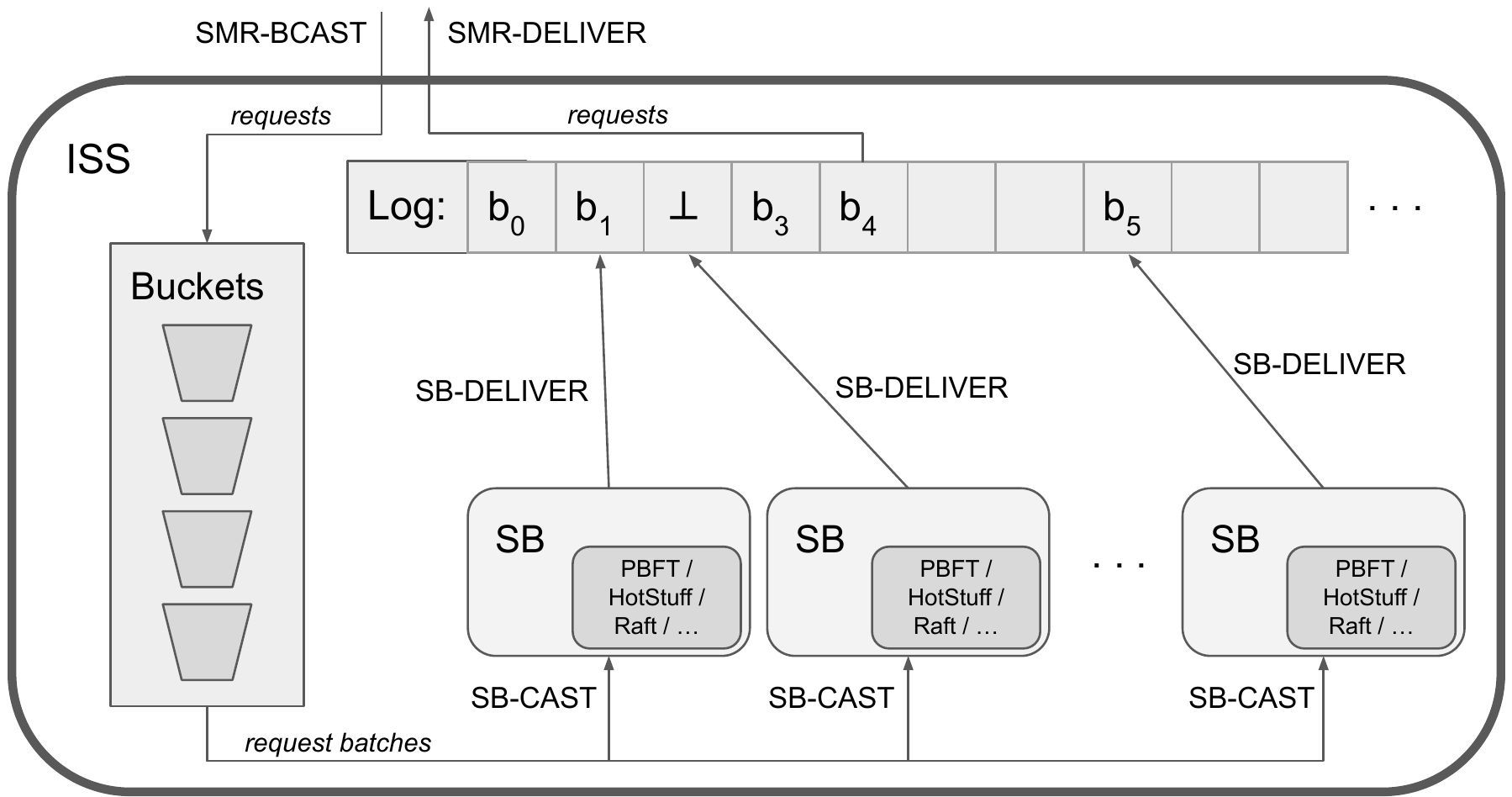}
\caption{\moduleAbbr invocation in \sysname.}
\label{fig:iss}
\end{figure}

\section{\sysname Algorithm Details}
\label{sec:details}

In this section we present the \sysname algorithm in detail.
The main high-level algorithm that produces a totally ordered log is described in \Cref{alg:top-level}.
For better readability, certain auxiliary functions and the functions related to epoch initialization
are presented separately in \Cref{alg:auxiliary,alg:epoch-init}, respectively.
We list the implementation of leader selection policies in \Cref{alg:leader-selection-policies}.
The notation $\langle i, e | args \rangle$ corresponds to an event $e$ of instance $i$ with arguments $args$.

As described in the previous sections, \sysname proceeds in epochs,
each epoch multiplexing multiple segments into a final totally ordered log.
We start with epoch number $0$ (line \ref{ln:init-epoch-0}) and an empty log (line \ref{ln:init-empty-log}).
All buckets are initially empty (line \ref{ln:init-empty-buckets}).
Whenever a client submits a new request (line \ref{ln:upon-request}),
\sysname adds the request to the corresponding bucket (line \ref{ln:request-to-bucket}).

We assume access accross all epochs to a module $D$ that implements an eventually strong failure detector, as defined in \Cref{sec:sb}.

\begin{algorithm}
  \caption{Main \sysname algorithm for node $p$}
  \scriptsize
  \label{alg:top-level}
  \begin{algorithmic}[1]
    \Implements
      \State State Machine Replication, \textbf{instance} \textit{smr}
    \EndImplements

    \Uses
      \State \modulename, \textbf{instance} \textit{\moduleAbbrSmall}($\sigma$, $S$, $M$, $d$) (multiple parametrized instances)
      \State with sender $\sigma$, sequence numbers $S$, messages $M$, failure detector \textbf{instance} $d$
      \State Eventually Strong Failure Detector \textbf{instance} $D$
    \EndUses

    \Parameters
      \State numBuckets
      \State batchTimeout
      \State epochLength
      \State maxBatchSize
    \EndParameters

    \Upon{init()}
    \State currentEpoch $\leftarrow$ 0\footnotemark
    \label{ln:init-epoch-0}
    \State segments $\leftarrow$ $\emptyset$
    \State firstUndelivered $\leftarrow 0$
    \State totalDelivered $\leftarrow 0$
    \ForAll{sn}{\mathbb{N}}
      \State log[$sn$] $\leftarrow$ $none$
      \label{ln:init-empty-log}
      \State proposed[$sn$] $\leftarrow$ $none$
      \EndForAll
      \ForAll{0 \leq b < numBuckets}{ }
        \State buckets[b] $\leftarrow$ empty bucket
        \label{ln:init-empty-buckets}
        \EndForAll
        \State initEpoch(0)
        \label{ln:init-init-0}
        \State runEpoch(0)
        \label{ln:init-run-0}
        \EndUpon

        \UponEvent{\langle smr, \bcast | req\rangle}
        \label{ln:upon-request}
        \If{valid($req$)}
          \State buckets[hash($req$)].add($req$)
          \label{ln:request-to-bucket}
        \EndIf
        \EndUponEvent

        \Function{runEpoch}{e}
          \ForAll{sn}{seqNrs(e) \textrm{ with } segOf(sn).leader = p}
            \State propose(sn)
            \label{ln:propose-call}
            \EndForAll
          \EndFunction

          \UponEvent{\langle segOf(sn).SB, \sbdeliver | (sn, batch)\rangle}
            \label{ln:module-deliver}
            \State log[$sn$] $\leftarrow batch$
            \label{ln:batch-to-log}
            \If{$batch \neq \bot$}
              \label{ln:deliver-batch}
              \ForAll{req}{batch}
                \State buckets[hash($req$)].remove($req$)
                \label{ln:remove-delivered-req}
              \EndForAll
            \ElsIf{proposed[$sn$] $\neq none$}
              \label{ln:rejected-batch}
              \State resurrectRequests(proposed[$sn$])
              \label{ln:resurrect-call}
            \EndIf
          \EndUponEvent

          \UponCondition{\forall sn \in seqNrs(\text{currentEpoch}) : \text{log}[sn] \neq none}
            \label{ln:epoch-advance-condition}
            \State initEpoch(currentEpoch + 1)
            \label{ln:init-next-epoch}
            \State runEpoch(currentEpoch + 1)
            \label{ln:run-next-epoch}
          \EndUponCondition

          \UponCondition{log[\text{firstUndelivered}] \neq none}
            \State deliver($sn$)
            \label{ln:in-order-delivery}
            \State firstUndelivered $\leftarrow sn+1$
          \EndUponCondition

            \algstore{main-pseudocode}
  \end{algorithmic}
\end{algorithm}
\footnotetext{We sloppily lists as sets and use set operators like $\cup$ or $\in$ where the meaning is clear from the context.}

\begin{algorithm}
  \caption{Auxiliary functions}
  \scriptsize
  \label{alg:auxiliary}
  \begin{algorithmic}[1]
    \algrestore{main-pseudocode}

    \Function{segOf}{sn}
      \State \textbf{return} $seg \in segments : sn \in seg.seqNrs$
    \EndFunction

    \Function{seqNrs}{e}
      \State \textbf{return} $\{i \in \mathbb{N} |$\\
      \hspace{4.5em} $e \cdot epochLength \leq i < e \cdot (epochLength + 1)\}$
    \EndFunction

    \Function{propose}{sn}
      \State \textbf{wait until} \hspace{1.5em} segments[sn].batchReady()\\
      \label{ln:propose-batch-full}
      \hspace{6em}$\OrT$batchTimeout\\
      \label{ln:propose-batch-timeout}%
      \label{ln:propose-catchup}
      \State batch $\leftarrow$ cutBatch(segOf($sn$).buckets)
      \label{ln:propose-cut-batch}
      \State \textbf{trigger} $\langle \text{seg.\moduleAbbr}, \sbcast | sn, batch \rangle$
      \label{ln:propose-broadcast}
      \State proposed[$sn$] $\leftarrow$ batch
      \label{ln:propose-track}
      \If{$batch \neq \bot$}
        \label{ln:propose-batch}
        \ForAll{req}{batch}
          \State buckets[hash($req$)].remove($req$)
          \label{ln:remove-proposed-req}
        \EndForAll
      \EndIf
    \EndFunction

    \Function{deliver}{sn}
      \State $batch\leftarrow\text{log}[sn]$
      \If{$batch\neq\bot$}
        \ForAll{req}{batch}
          \State \textbf{trigger} $\langle smr, \deliver | req, \text{totalDelivered} \rangle$
          \label{ln:smr-delivery}
          \State totalDelivered $\leftarrow$ totalDelivered $+1$
        \EndForAll
      \EndIf
    \EndFunction

    \Function{cutBatch}{\mathcal{B}}
       \State \textbf{return} a batch of the $maxBatchSize$ oldest requests in $\mathcal{B}$
       \label{ln:oldest}
    \EndFunction

    \Function{resurrectRequests}{batch}
      \ForAll{req}{batch}
        \State buckets[hash(req)].add(req)
        \label{ln:resurrect-add}
        \EndForAll
      \EndFunction

      \algstore{main-pseudocode}
  \end{algorithmic}
\end{algorithm}

\subsection{Epoch Initialization}

At the start of each epoch $e$, \sysname:
  (1) calculates $\epLeaders{e}$, the set of nodes that will act as leaders in $e$,
  based on the used leader selection policy (line \ref{ln:init-leaders}),
  (2) for each node $l$ in $\epLeaders{e}$, creates a new segment with leader $l$ (lines \ref{ln:init-new-segment} and \ref{ln:init-assign-leader}),
  (3) assigns all sequence numbers $Sn(e)$ of epoch $e$ in a round-robin way to all created segments (line \ref{ln:init-assign-sns}),
  (4) assigns buckets to the created segments as described in \Cref{sec:req-dupl} (line \ref{ln:init-assign-buckets}), and
  (5) creates an instance of \moduleAbbr for each created segment (line \ref{ln:init-module}).

  \begin{algorithm}
    \caption{Epoch initialization}
    \scriptsize
    \label{alg:epoch-init}
    \begin{algorithmic}[1]
      \algrestore{main-pseudocode}

      \Function{initEpoch}{e}
        \State currentEpoch $\leftarrow e$
        \State
        \State leaders $\leftarrow$ leader\_selection\_policy($e$)
        \label{ln:init-leaders}
        \ForAll{0 \leq l < |leaders|}{ }
          \State seg $\leftarrow$ new segment
          \label{ln:init-new-segment}
          \State seg.leader $\leftarrow$ leaders[$l$]
          \label{ln:init-assign-leader}
          \State seg.seqNrs $\leftarrow \{sn \in seqNrs(e) \hspace{.4em} |$\\
                 \hspace{9em} $sn \equiv l \textrm{ mod } |leaders|\}$
          \label{ln:init-assign-sns}
          \State seg.buckets $\leftarrow$ Buckets($e, leaders[l]$)
          \label{ln:init-assign-buckets}
          \State seg.\moduleAbbr $\leftarrow$ \moduleAbbrSmall(leaders[l], batches(seg.buckets), seg.seqNrs, D)
          \State \textbf{trigger} $\langle seg.\moduleAbbr, \sbinit\rangle$
          \label{ln:init-module}
          \State segments $\leftarrow$ segments $\cup$ $\{$seg$\}$
          \EndForAll
        \EndFunction

        \Function{Buckets}{e, i}
          \ForAll{n}{\mathcal{N}}
            \State initBuckets[n] $\leftarrow \{b < numBuckets$ $|$\\
                          \hspace{10.4em} $(e + b) \equiv n \textrm{ mod } |\mathcal{N}|\}$
            \EndForAll
            \State leaders $\leftarrow$ epLeaders\_POLICY($e$)
            \State extraBuckets $\leftarrow \{b : \exists n \hspace{.4em} |$\\
                   \hspace{8.5em} $n \notin leaders \wedge b \in initBuckets[n]\}$
            \State \textbf{return} $initBuckets[i] $ $\cup$ $\{b \in extraBuckets \hspace{.4em} |$\\
                   \hfill $(e + b) \equiv (\textrm{index of } i \textrm{ in } leaders) \textrm{ mod } |leaders|\}$
          \EndFunction

          \algstore{main-pseudocode}

    \end{algorithmic}
  \end{algorithm}

\subsection{Ordering Request Batches}
\label{sec:ordered-log}

\sysname orders requests in batches, a common technique which allows requests to be handled in parallel, which amortizes the processing cost of protocol messages, and, thereby, improves throughput.
During an epoch, every node $l$ that is the leader of a segment $s$
proposes request batches for sequence numbers assigned to $s$ (line \ref{ln:propose-call}).
$l$ does so by sb-casting the batches using the instance of \moduleAbbr associated with $s$.
Every node then inserts the sb-delivered (sequence number, batch) pairs
at the corresponding positions in its copy of the log (line \ref{ln:batch-to-log}).
We say the node \emph{commits} the batch with the corresponding sequence number since, once inserted to the log, the assignment of a batch to a sequence number is final.

\textbf{Proposing Batches.}
Each node maintains local data structures of \emph{buckets queues}, which store the received and not yet proposed or delivered requests corresponding to the respective bucket.
To propose a request batch for sequence number $sn$,
$l$ first constructs the batch using requests in the bucket queues corresponding to the buckets assigned to $s$.
To implement efficient request batching while preserving low latency,
$l$ waits until at least one of the following conditions is fulfilled:

\begin{itemize}
  \item The bucket queues corresponding to $s$ contain enough requests (more than a predefined $batchSize$) (line \ref{ln:propose-batch-full}).
  \noindent\item A predefined time elapses since the last proposal (line \ref{ln:propose-batch-timeout}).
  Under low load, this condition sets an upper bound on the pending%
  \footnote{Latency from the moment a request is received until it is added in a batch.}%
  latency of requests waiting to be proposed, even if the batch is filling up slowly.
\end{itemize}

$l$ then constructs a $batch$ using up to $batchSize$ requests
(line \ref{ln:propose-cut-batch}),
removes those requests from their bucket queues
and proposes the batch by invoking $\sbcast(sn, batch)$ on the \moduleAbbr instance associated with $s$ (line \ref{ln:propose-broadcast}).

Every leader also keeps track of the values it proposed for each sequence number (line \ref{ln:propose-track}).
This, as we explain later, is important in the case of asynchrony.

\textbf{Assembling the Final Log.}
Whenever any instance of \moduleAbbr (belonging to any segment) delivers
a value associated with a sequence number (line \ref{ln:module-deliver}) at some node $n$,
$n$ inserts the delivered value at position $sn$ of the log (line \ref{ln:batch-to-log}).
If a request batch $\neq\bot$ has been delivered (line \ref{ln:deliver-batch}),
$n$ removes the contained requests from their corresponding bucket queues (line \ref{ln:remove-delivered-req})
to avoid proposing them again in a later epoch.

Note that bucket queues are local data structures at each node, and thus each node manages its bucket queues locally.
Nodes add all requests they obtain from clients to their local bucket queues,
but only propose batches constructed from the queues corresponding to the buckets assigned to their segments.
Consiquently, a node $n$ delivers batches containing requests mapping to other buckets than those $n$ uses for proposing.
Therefore, to avoid request duplication across epochs, each node must remove all delivered requests from its local bucket queues.

If the special value $\bot$ has been delivered by \moduleAbbr
and, at the same time, $n$ itself had been the leader proposing a batch for $sn$ (line \ref{ln:rejected-batch}),
$n$ ``resurrects'' all requests in the batch it had proposed (line \ref{ln:resurrect-call})
by returning them to their corresponding bucket queues (line \ref{ln:resurrect-add}).
This scenario can appear in the case of asynchrony / partitions, where a correct leader is suspected as faulty after having proposed a batch.
Such a leader must return the unsuccessfully proposed requests in their bucket queues and,
if batches with those requests are not committed by other nodes in the meantime,
retry proposing them in a later epoch where it is again leader of a segment with those buckets.

A node considers the ordering of a request finished when it is part of a committed batch with an assigned sequence number $sn$ and the log contains an entry for each sequence number $sn' \leq sn$.

Each request is delivered with a unique sequence number $sn_r$ denoting the total order of the request.
$sn_r$ is derived from the sequence number of the batch in which the request is delivered
and the position of the request in the batch.
Let $\mathcal{S}_{sn}$ be the number of requests in a batch delivered with sequence number $sn$
and let $r$ be the $k^{th}$ request in this batch.
For each such request $r$, \sysname outputs $\deliver(r,sn_r)$ where:
\begin{equation}\label{eq:requests}
  sn_r= k+\sum_{i=0}^{sn-1}\mathcal{S}_{i}
\end{equation}

\subsection{Advancing Epochs}
\label{sec:epochs}

\sysname advances from epoch $e$ to epoch $e+1$ when the log contains an entry for each sequence number in $Sn(e)$ (line \ref{ln:epoch-advance-condition}).
This will eventually happen for each epoch at each correct node due to \moduleAbbr \textit{Termination}.
Only then does the node start processing messages related to epoch $e+1$
and starts proposing batches for sequence numbers in $Sn(e+1)$
(lines \ref{ln:init-next-epoch} and \ref{ln:run-next-epoch}).

Requiring a node to have committed all batches in epoch $e$ before proposing batches for $e+1$
prevents request duplication across epochs.
When a node transitions from $e$ to $e+1$, no requests are ``in flight''---%
each request has either already been committed in $e$ or has not yet been proposed in $e+1$.

\subsection{Selecting Epoch Leaders}
\label{sec:leader-selection}

In order to guarantee that each request $r$ submitted by a correct client is ordered (liveness),
we must ensure that, eventually, there will be a segment in which $r$ is committed.
As implied by the specification of \moduleAbbr, this can only be guaranteed if a correct leader $p$ proposes a batch containing $r$
and the failure detector does not suspect $p$ until $r$ is committed.
The choice of epoch leaders is thus crucial.

\sysname selects leaders according to a \emph{leader selection policy},
a function known to all nodes that, at the end of each epoch $e$,
determines the set of leaders for epoch $(e + 1)$.

In this work, we describe leader selection policies that are evaluated locally, such that each node need only use, in a deterministic fashion, information guaranteed to be present at all nodes.
Evaluation of the policy can, in principle, involve communication among the nodes,
but we focus on simpler, locally evaluable policies that only take into account the epoch number $e$ and the state of the log up to $\maxSnE{e}$.

In order  to guarantee liveness of the system, the leader selection policy must ensure, for each bucket $b$,
that, in an infinite execution, $b$ will be assigned infinitely many times to a segment with a correct leader
that is not suspected by the failure detector.
Weak eventual accuracy (see \Cref{sec:sb}) guarantees that there exists such a leader.

Many different leader selection policies ensure liveness, while posing different trade-offs with respect to performance.
In the following, we present a few example policies.

\paragraph{SIMPLE}
The SIMPLE leader selection policy always selects all nodes to be leaders in each epoch.
Faulty leaders might inhibit progress of the requests which map to the buckets assigned to their segments,
but all correct nodes stay in the leaderset forever, and thus, thanks to bucket re-assignment,
each bucket will be assigned to all correct leaders infinitely many times.

This policy results in always fully utilizing all correct nodes' resources, which optimizes the maximal system throughput at saturation in a fault-free execution.
When, however, one or more nodes are faulty, they may delay epoch transitions, even after the correct leaders finish assigning batches to all sequence numbers in their corresponding segments.
Thus, the SIMPLE policy can have negative impact on request latency and on average throughput in presence of node failures.
If the correct nodes have enough resources to handle all the incoming requests,
it is beneficial to exclude misbehaving nodes from the leaderset to prevent them from unnecessarily delaying requests.

\paragraph{BACKOFF}
The BACKOFF leader selection policy is similar to the SIMPLE policy
in the sense that all suspected nodes are eventually re-included in the leaderset.
Instead of re-including $i$ immediately in the next epoch, the BACKOFF policy bans $i$ from the leaderset for a certain number of epochs.
The ban period is doubled each time $i$ is suspected and decreases linearly when $i$ behaves correctly for a whole epoch after its re-inclusion.
Liveness is ensured by the same principle as with the SIMPLE policy,
as even with BACKOFF, all correct nodes will eventually be included in the leaderset forever.

The BACKOFF policy improves on the latency of SIMPLE policy,
but pays the price of sub-optimal throughput in presence of transient problems that lead to suspecting correct leaders.
The BACKOFF policy can exclude a large number of nodes, even more than $f$, in which case at least some of them must be correct.
In the extreme case, the BACKOFF policy can result in epochs where all leaders are banned (in this corner case, \sysname simply skips such epochs).%
\footnote{Skipping epochs that have no leaders is not reflected in the pseudocode for simplicity of presentation.}

\paragraph{BLACKLIST}
The BLACKLIST leader selection policy never excludes more than $f$ nodes from the leaderset.
It maintains a blacklist of up to $f$ most recently suspected nodes and uses the remaining (at least) $2f + 1$ nodes as leaders.
It does not guarantee that all correct nodes will eventually become leaders forever
but as long as this is the case for at least one correct node, liveness is still preserved.
With BLACKLIST, the leaderset contains at least $f + 1$ correct nodes in each epoch.
The BLACKLIST policy corresponds to the blacklisting mechanism used in BFT-Mencius~\cite{BFT-Mencius}.

BLACKLIST is the asymptotically optimal policy for the case where all the $f$ \emph{tolerated} node failures actually manifest.
In such a case, all faulty nodes will be eventually excluded forever, and all correct nodes will be leaders.
However, in less severe failure scenarios, where less than $f$ nodes are \emph{actually faulty},
correct nodes might be excluded from the leaderset forever only due to transient network asynchrony / partitions.

This list of policies is not exhaustive.
Moreover, the discussion above hints that there is no policy that is optimal for all scenarios.
In the rest of this work, the reader should assume the BLACKLIST policy as it is the \emph{default} \sysname leader policy throughout the paper unless mentioned otherwise.

\begin{algorithm}
  \caption{Leader selection policies}
  \scriptsize
  \label{alg:leader-selection-policies}
  \begin{algorithmic}[1]
    \algrestore{main-pseudocode}

    \Parameters
      \State banPeriod
      \State c
      \Comment{constant decrease of ban period}
    \EndParameters

    \Function{epLeaders\textrm{\_SIMPLE}}{e}
      \State \textbf{return} sorted($\mathcal{N}$)
    \EndFunction

    \Function{epLeaders\textrm{\_BLACKLIST}}{e}
      \State failures $\leftarrow \emptyset$
      \ForAll{n}{\mathcal{N}}
        \State failures $\leftarrow lastFailure(n,e)$
      \EndForAll

        \State leaderCount $\leftarrow 0$
        \ForAll{n}{\mathcal{N}}
          \If{$lastFailure(n,e)$ not among the $f$ highest values in failures}
            \State leaders[$leaderCount$] $\leftarrow n$
            \State leaderCount $\leftarrow$ leaderCount + 1
          \EndIf
          \EndForAll
          \State \textbf{return} leaders
      \EndFunction

        \UponCondition{\forall sn \in seqNrs(\text{e}) : \text{log}[sn] \neq none}
        \Comment{Epoch e finished}
          \ForAll{n}{\mathcal{N}}
            \If{$suspect(n,e)$}
              \If{penalty[n]$>0$}
                \State penalty[n] = penalty[n]*$2-1$
              \Else
                \State penalty[n] = banPeriod
              \EndIf
            \Else
              \If{penalty[n]$>0$}
                \State penalty[n] = penalty[n]$-c$
              \EndIf
            \EndIf
          \EndForAll
        \EndUponCondition

        \Function{epLeaders\textrm{\_BACKOFF}}{e}
          \State leaderCount $\leftarrow 0$
          \ForAll{n}{\mathcal{N}}
            \If{penalty[n]$\le0$}
              \State leaders[$leaderCount$] $\leftarrow n$
              \State leaderCount $\leftarrow$ leaderCount + 1
            \EndIf
          \EndForAll
          \State \textbf{return} leaders
        \EndFunction

        \Function{lastFailure}{n,e}
        \Comment{Returns the highest seq no node $n$ failed to deliver}
          \State failures $\leftarrow \{sn \leq \max(seqNrs(e-1)): segOf(sn).leader = n$ $\wedge$ $\log[sn] = \bot\}$
          \If {$failures = \emptyset$}
            \State \textbf{return $-1$}
          \EndIf
          \State \textbf{return} max($failures$)
        \EndFunction

        \Function{suspect}{n,e}
          \State \textbf{return} $lastFailure(n,e) \in seqNrs(e)$
        \EndFunction
  \end{algorithmic}
\end{algorithm}

\subsection{Checkpointing and State Transfer}
\label{sec:checkpoint}
\sysname implements a simple checkpointing protocol.
Every node $i$, in each epoch $e$, when the log contains an entry for each sequence number in $Sn(e)$, broadcasts a signed message \\
$\langle\checkpoint,\linebreak[1] max(Sn(e)),\linebreak[1] D(e),\linebreak[1] \sigma_i \rangle$, where $D(e)$ is the Merkle tree root of the digests of all the batches in the log with sequence numbers in $Sn(e)$.
Upon acquiring a strong quorum of $2f+1$ matching \checkpoint messages with a valid signature against the sender node's public key, node $i$ creates a stable checkpoint\\
 $\langle\stableCheckpoint,\linebreak[1] max(Sn(e)),\linebreak[1] \pi(e) \rangle$, where $\pi(e)$ is the set of $2f+1$ signatures on the \checkpoint messages.
At this point, $i$ can garbage collect all segments of epoch $e$.

When a node $i$ has fallen behind, for example when the node starts receiving messages for a future epoch, the node performs a state transfer, i.e., $i$  fetches the missing log entries along with their corresponding stable checkpoint which proves the integrity of the data.

\sysname checkpointing is orthogonal to any checkpointing and state transfer mechanism pertaining to the \moduleAbbr implementation because \moduleAbbr instances must terminate independently.

\subsection{Membership Reconfiguration}
\label{sec:reconfiguration}
A detailed membership reconfiguration protocol is outside the scope of this paper.
However, we outline a solution.
Thanks to \moduleProperty~3 (Termination), all correct nodes eventually deliver a value for each sequence number of an epoch.
Moreover, thanks to \smrProperty~2 (Agreement), all correct nodes assemble the same log at the end of the epoch.
Therefore, the log at the end of the epoch can be used to deterministically make decisions for the next epoch, including decisions about nodes and clients joining/leaving the set of system processes.
Such a decision can be based, for example, on a flagged reconfiguration request proposed by a manager process~\cite{Reiter96membership} which becomes part of the log.

\subsection{Request Handling}
\label{sec:requests}

As we defined in \Cref{sec:model}, a request $r=(o,id)$ with payload $p$ and identifier $id=(t,c)$ is wrapped in a message $\langle \request, r \rangle_{\sigma_c}$.
Our implementation represents the client identifier $r.id.c$ with an integer which it associates to the client's $c$ public key.
The signature is calculated over the request identity $r.id$ and payload $r.o$ to guarantee integrity and authenticity.

Similarly to Mir-BFT~\cite{stathakopoulou2019mir}, clients can submit multiple requests in parallel within a client watermark window, i.e., a contiguous set for the per-client request sequence number $r.id.t$.
\sysname advances all clients' watermark widows at the end of each epoch.

Each correct node,  upon receiving a valid request, adds the request, based on its identifier, to the corresponding bucket queue.
A request is considered \emph{valid} if:
(1) it has a valid signature
(2) the public key corresponds to a client in the client set $C$ of the system, and
(3) is within the client watermarks.
Bucket queues are idempotent, i.e., each correct node adds a request to the corresponding bucket queue exactly once.
Moreover, the bucket queue implementation maintains a FIFO property to guarantee liveness with the oldest request always being proposed first.

Requests are uniformly distributed to buckets using modulo as a hash function.
With $|\mathcal{B}|$ denoting the total number of buckets and $||$ denoting concatenation, each request $r$  maps to a bucket $b$:
\begin{equation*}
  b = r.id.c || r.id.t \textrm{ mod } |\mathcal{B}|
\end{equation*}

Importantly, we exclude the payload of the request from the bucket mapping function
to prevent malicious clients from biasing the uniform distribution.
In a permissioned system, the client cannot assume different identities and may only bias the outcome of the hash function by choosing the request sequence number.
However, we limit the available sequence numbers for each client, and therefore their ability to bias the request distribution to buckets, with the client watermarking mechanism.

\noindent\textbf{Request execution.}
\sysname is oblivious to the payload of the requests for general applicability.
Execution is not part of \sysname; however, it can be coupled with any application that requires a total order of requests.
Moreover, execution against a state machine is straight-forward.
A request, as part of a batch, is considered part of the log (and can be therefore executed) once all previous batches are added to the log (see \Cref{sec:ordered-log}).
Therefore, a request can be executed against the state machine once it is added to the log.
This does not require the epoch, in which the request is added to the log, to finish.

\section{\sysname Implementation}
\label{sec:implementation}

We implement \sysname in Go, using gRPC for communication with TLS on
nodes with two network interfaces: one for client-to-node
and one for node-to-node communication.

Our implementation is highly concurrent: multiple threads are handling incoming client requests,
verifying request signatures, sending / receiving messages, and executing various sub-protocols
such as checkpointing and fetching missing protocol state.
Each \moduleAbbr instance also executes in its own thread.
A separate thread orchestrates all of the above.

In the rest of this section,
we discuss the operating principle of \sysname{}s implementation and
the interface between \sysname and the \moduleAbbr implementation (\ref{sec:impl-interface}),
our implementation of the ordering subprotocols (\ref{sec:impl-protocols}),
the implementation of the failure detector module (\ref{sec:impl-fd}),
the interaction between \sysname and its clients (\ref{sec:impl-clients}),
and crucial technical aspects for achieving robustness and high performance (\ref{sec:impl-misc}).

\subsection{Operating Principle and Interface to \moduleAbbr}
\label{sec:impl-interface}

Our \sysname implementation is composed of multiple modules.
This section will focus on \sysname' core modules of interest, the \emph{Manager} and the \emph{Orderer},
whose interaction drives the system's operation.

\textbf{The Manager} module implements the high-level logic of epochs, segments, leaders, and buckets.
Based on the state of the log (i.e., batches that have been committed),
the Manager advances between epochs and, in each epoch,
creates that epoch's segments, assigns leaders and buckets to them, and announces them to the Orderer.

\textbf{The Orderer} module processes segments announced by the Manager.
For each segment, the Orderer instantiates an implementation of the \moduleAbbr protocol
parametrized according to the segment.
We integrate different ordering protocols
by providing different implementations of the Orderer module.

A simple interface between the Orderer and the rest of the system encapsulates the ordering protocol.
The Orderer is required to implement a \emph{Segment(s)} operation that can be called by the Manger.
When the Manager issues a segment $s$, it invokes \emph{Segment(s)} at the Orderer.
$s$ contains all information relevant to the segment such as the set of sequence numbers,
the associated buckets, and the leader.
On reception of the segment, the Orderer creates an instance of the ordering protocol
responsible for assigning request batches to all sequence numbers of $s$.

\sysname provides an \emph{Announce(b, sn)} operation to the Orderer.
The Orderer invokes it to announce a request batch $b$ committed at sequence number $sn$.
After the Manager invokes \emph{Segment(s)} at the Orderer,
the Orderer's only responsibility is to invoke $Announce(b, sn)$
exactly once for each sequence number $sn$ in $s$,
with $b$ consisting only of (unique) requests from buckets assigned to $s$.

Our \sysname implementation augments this core \emph{Segment(s)}, \emph{Announce(b, sn)} interface by
operations for interaction with other modules
such as sending and receiving messages, state transfer or parallel request signature verification.
We omit those technical details for clarity.


\subsection{\moduleAbbr Implementation}
\label{sec:impl-protocols}

In this section we discuss (1) how we implement \moduleAbbr with different leader-driven consensus protocols (PBFT, HotStuff, Raft)  and (2) adaptations of those protocols critical for \sysname performance.

All protocol implementations adhere to the following common design principles:

\noindent1.We initialize the protocol such that the the first protocol leader is the segment leader (dedicated \moduleAbbr sender) and all other nodes of the system participate as followers. \\
\noindent2. After a leader-change, any new leader (including the segment leader if it becomes leader again),
only proposes $\bot$ values for any sequence number not initially proposed by the segment leader.%
\footnote{Enforcing $\bot$ values is necessary for \moduleAbbr to be implementable. Otherwise, both \moduleAbbr \textit{Integrity} and \moduleAbbr \textit{Termination} cannot be satisfied at the same time.} \\
\noindent3. A follower accepts a proposal only if
(a) all requests in the batch are valid according to~\Cref{sec:requests},
(b) no request in the batch has previously been proposed in the same epoch or committed in a previous epoch
(c) all requests belong to the buckets of the segment,
and (d) either the segment leader sb-casts the proposal, or the proposal is $\bot$.

Our implementation currently supports two Byzantine fault-tolerant total order broadcast protocols (PBFT and HotStuff) and a crash fault-tolerant one (Raft).

\subsubsection{PBFT}

We follow the PBFT protocol as described by Castro and Liskov~\cite{Castro:2002:PBF}, with a few adaptations.

Our implementation need not deal with timeouts at the granularity of single requests, as PBFT does.
To prevent censoring attacks (and thus ensure liveness), a PBFT replica initiates a view change
if any request has not been committed for too long.
Since \sysname prevents censoring attacks by bucket re-assignment,
it is sufficient for us to make sure to commit \emph{some batch} before a timeout fires
and reset this timer when committing any batch.
In the absence of incoming requests, the primary periodically proposes an empty batch to prevent a view change.
Moreover, for simplicity, we implement view-change with signatures according to Castro and Liskov~\cite{castro1998practical}.

\subsubsection{HotStuff}
We implement chained HotStuff according to Yin \textit{et. al}~\cite{hotstuff} with BLS~\cite{boneh2001short} threshold signatures using DEDIS library for Go~\cite{dedis}.

In our implementation each batch corresponds to a HotStuff command, and each segment sequence number corresponds to a HotStuff view.
Each segment is implemented as a new HotStuff instance with a new root certificate $QC_0$.
To ensure that all sequence numbers can be delivered, i.e., to ensure that we can always ``flush'' the pipeline of chained HotStuff, we  extend the segment with 3 dummy sequence numbers corresponding to dummy empty batches which are not added to the \sysname log.
\Cref{fig:hotstuff} demonstrates an example of a segment with 3 sequence numbers.

\begin{figure}[h]
\centering
  \includegraphics[width=\columnwidth]{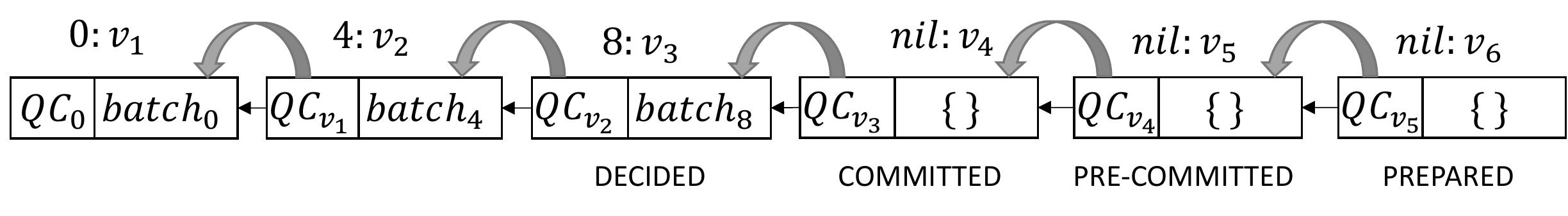}
  \caption{\label{fig:hotstuff}
  Chained HotStuff execution for a segment with 3 sequence numbers $\{0,4,8\}$. When view $v_6$ is prepared, $v_3$ is decided and $batch_8$  can be added to the log.
  $QC_{v_i}$ is the HotStuff quorum certificate for view $i$ (e.g., a threshold signature from $2f+1$ followers) on the proposal of view $v_i$.
  }
\end{figure}

\subsubsection{Raft}

Briefly, in Raft nodes set a random timer within a configurable range, which they reset every time they receive a message from the elected leader.
If the timer fires, the node advances to a new term (similar to PBFT view) and enters an election phase as a candidate leader.
An elected leader periodically sends append-entry requests for new values, possibly also empty, as a heartbeat.
The leader collects responses according to which it might resend to the followers any value they declare as missing.

We implement Raft according to~\cite{RAFT} with minor adaptations.
We fix the first leader to be the leader of the segment skipping the election phase.
Until the end of the segment, the leader periodically sends append-entry requests containing batches.
The leader continues to send empty append-entry requests until the end of the epoch to guarantee that enough nodes have added all the batches of the segment to their log.

Raft~\cite{RAFT} authors suggesst a fixed range for the random timer, based on empirical data.
In our implementation, to ensure liveness under the assumption of eventual synchrony, following the approach of PBFT,
we increase the time window for leader election.
In particular, in case  a node fails to elect a leader before advancing to a new term we double the minimum and maximun value of the random timer range.

  \subsubsection{Failure Detector Implementation}
  \label{sec:impl-fd}
  So far we have assumed that $\moduleAbbr$ has access to a failure detector module
  as defined in \Cref{sec:sb}.
  In practice, we extract the failure detector functionality from the underlying
  protocols we use to implement $\moduleAbbr$.
  All the aforementioned protocol implementations support a failure detection mechanism
  based on timeouts.

  In particular, in PBFT nodes suspect the primary (leader), if some request is not committed before a timeout.
  The timeout increases if no progress can be made.
  Under the assumption of synchrony after GST, the timeout will become longer than the network delay, and, therefore,
  a quiet primary will be suspected by every correct node (satisfying $\diamond S(bz)$ Stong Completeness)
  and no correct primary will be suspected by any correct node (satisfying $\diamond S(bz)$ Eventual Weak Accuracy).
  Similarly, in HotStuff the $\diamond S(bz)$ properties are satisfied by the Pacemaker mechanism.
  Finally, our Raft implementation satisfies $\diamond S(bz)$ properties by increasing the the time window for leader election.

\subsection{Interaction with Clients}
\label{sec:impl-clients}

Clients submit requests to \sysname by sending request messages to nodes.
To guarantee Liveness (\Cref{sec:model}, \smrProperty4) a client must ensure that at least one correct leader eventually receives the request.
A trivial solution is to send the request to all nodes.
However, \sysname implements an optimistic a \emph{leader detection} mechanism to (1) help the clients find the correct leader for each request faster and (2) better load balance request processing among the nodes.

At any point in time, the bucket to which a particular request belongs is assigned to a single segment with a unique leader.
Therefore, the client only needs to send its request to the node currently serving as a leader for the corresponding bucket.

When a node delivers a request $r$, as described in~\Cref{sec:ordered-log}, it sends a response message to the client that submitted $r$.
When the client obtains a quorum of responses, it considers the request delivered.

\sysname keeps the clients updated about the assignment of buckets to leaders.
At each epoch transition, all  nodes send a message with the assignment for the next epoch to all clients.
A client accepts such a message once it receives it from a quorum of nodes.
The client submits all subsequent requests for this epoch to the appropriate node.
Moreover, it resubmits all requests submitted in the past that have not yet been delivered.
This guarantees that all correct nodes will eventually receive the request, which ensures liveness.

To make sure that, in most cases, a leader already has a request when it is that leader's turn to propose it,
the client sends its request to two additional nodes
that it projects to be assigned the corresponding bucket in the next two epochs.


\subsection{Important Technical Aspects}
\label{sec:impl-misc}

We now describe several aspects of our implementation that
are fundamental for ensuring \sysname{}s robust performance.

\subsubsection{Rate-limiting Proposals in PBFT}
PBFT's ability to send proposals in parallel is instrumental for achieving high throughput.
However, as soon as a load spike or a temporary glitch in network connectivity occurs
(as it regularly does on the used cloud platform),
the leader can end up trying to send too many batches in parallel.
Due to limited aggregate bandwidth, \emph{all} those batches will take longer to transfer,
triggering view change timeouts at the replicas.

We address this issue by setting a hard limit on the rate of sending batches ``on the wire'',
allowing (the most part of) a batch to be transmitted before the transfer of the next batch starts.
This measure limits peak throughput but is effective at protecting against unnecessary view changes.

\subsubsection{Concurrency Handling}

A naive approach to handling requests,
where each client connection is served by a thread that, in a loop,
receives a request, locks the corresponding bucket queue, adds the request, and unlocks the bucket queue,
is detrimental to performance. We attribute this to cache contention on the bucket queue locks.

Still, access to a bucket does have to be synchronized, as adding (on reception) and removing (on commit) requests
must happen atomically.
At the same time, an efficient lock-free implementation of a non-trivial data structure such as a bucket queue
could be a research topic on its own.

We thus dedicate a constant, limited number of threads (as many as there are physical CPU cores)
to only adding requests to bucket queues, such that each bucket queue is only accessed by one thread,
removing most of the contention.
The network-handling threads pass the received requests to the corresponding bucket-adding threads
using a lock-free data structure optimized for this purpose (a Go channel in our case).

\subsubsection{Deployment, Profiling, and Analysis}

\sysname comes with tools for automating the deployment of hundreds of experiments
across hundreds of geo-distributed nodes on the cloud and for analyzing their outputs.
They profile (using \texttt{pprof}) the execution at each node,
pinpointing lines of code that cause stalling or high CPU load.
For example, the above-mentioned cache contention was pointed to by the profiler.
They also plot various metrics over time,
such as the size of proposed batches, commit rate, or CPU load.
Automatic exploration of the multi-dimensional parameter space
proved essential for understanding the inner workings of the system.

Using hundreds of cloud machines with hourly billing
also incurs significant cost.
Automatically commissioning cloud machines only for the time strictly necessary to run our experiments
and releasing those resources as soon as possible most likely saved thousands of dollars.

\section{Correctness}
\label{sec:correctness}
\subsection{From Consensus to \moduleAbbr}
\label{sec:sb-proof}
In this section we prove that we can implement \moduleAbbr with consensus in our system model (\Cref{sec:model}).

In \Cref{alg:sb-impl} we implement \moduleAbbr using Byzantine reliable broadcast (BRB),
$\Diamond S(bz)$ failure detector~\cite{malkhi1997unreliable} and Byzantine consensus (BC).
In \Cref{sec:sbz-impl} we outline how $\Diamond S(bz)$ can be implemented using timeouts and BRB.
BRB can be implemented under asynchronous communication assumptions~\cite{Bracha:1985:ACB:4221.214134}, unlike BC~\cite{FLP}.
BRB is, therefore, weaker than BC.
Therefore, \moduleAbbr is not stronger than (can be implemented using) BC.

For completeness we repeat below Byzantine reliable broadcast and Byzantine consensus definitions, the properties of which we then use to prove \moduleAbbr properties.
Our definitions follow that of Cachin et \textit{al.}~\cite{CachinGR11}.

\subsubsection{Byzantine Reliable Broadcast}

A process (node) $p$ reliably broadcasts a message (we write brb-casts), including including itself, by triggering
$\langle \rbcast | m \rangle$ and all correct process brb-deliver the message by triggering an event $\langle \rbdeliver | p, m \rangle$ such that the following properties hold:

\noindent \textbf{\rbProperty1 No duplication:}
Every correct process brb-delivers at most one message.\\
\noindent \textbf{\rbProperty2 Integrity: }
If some correct process delivers a message $m$ with sender $p$ and process $p$ is correct,
then $m$ was previously broadcast by $p$.\\
\noindent \textbf{\rbProperty4 Validity:}
If a correct process $p$ brb-casts a message $m$, then every correct process eventually brb-delivers $m$.\\
\noindent \textbf{\rbProperty5 Consistency:}
If some correct process $p$ brb-delivers a message $m$  and some correct process $q$ brb-delivers a message $m'$, then $m=m'$.\\
\noindent \textbf{\rbProperty6 Totality:}
If some message is delivered by any correct process, then every correct process eventually delivers a message.

\subsubsection{Byzantine Consensus}
A set of processes (nodes) use consensus to decide on a common value out of values they initially propose.

Formally, we write $\langle \bcpropose | v \rangle$ to denote that a process proposes value $v$
and $\langle \bcdecide | v \rangle $ to denote that a process outputs a decided value $v$.
Consensus in the BFT model (BC) guarantees the following properties:

\noindent \textbf{\bcProperty1 Termination: }
Every correct process decides on some value.\\
\noindent \textbf{\bcProperty2 Agreement:}
All correct processes decide on the same value.\\
\noindent \textbf{\bcProperty3 Integrity:}
No correct process decides twice.\\
\noindent \textbf{\bcProperty4 Validity:}
If all correct processes propose the same value $v$, then they decide on $v$;
otherwise, a correct process may only decide on a value that was proposed by some correct process
or the special value $\bot$.

$\bot$ indicates that the processes could not agree on any of the proposed values.
This could occur if not all correct process propose the same value, allowing correct process to decide on some value even if not all processes are correct (Strong Validity).

\subsubsection{$\Diamond S(bz)$ Implementation}
\label{sec:sbz-impl}
We assume that all correct processes maintain a timer for all other processes.
Moreover, each correct process $i$ periodically brb-casts a heartbeat message
$\langle \rbcast | h \rangle$.
Upon an event $\langle \rbdeliver | j, h \rangle$, a correct process $i$ restarts the timer for a process $j$ and triggers a $\langle\restore | j \rangle$ event if $j$ is in the list of suspects.
When the timer for process $j$ expires, a correct process $i$ triggers an event $\langle\suspect | j \rangle$ to add $j$ in the list of suspects if it is not there yet and doubles the timeout period for $j$.

We now show that this algorithm implements $\Diamond S(bz)$ with BRB in a partially synchronous system.
\paragraph{Strong Completeness.}
A quiet node $j$ by definition does not brb-cast messages, therefore, a correct node $i$ permanently suspects $j$ after the first timer for $j$ expires.

\paragraph{Eventual Weak Accuracy.}
After GST we assume that the communication is synchronous.
I.e., there is a bound $\Delta$ on network delay.
Therefore, the timeout period for any correct node $j$ will eventually be long enough, such that all correct nodes always brb-deliver $j$'s  heartbeats.

\subsubsection{\moduleAbbr Implementation}
The implementation of $\moduleAbbr(\sigma, S, M, D)$ is as follows.

Whenever an event $\langle \sbcast | sn, m \rangle$ is triggered, the dedicated sender $\sigma$ brb-casts a message form $m$ with $sn$.
All correct nodes run consensus for each sequence number in $S$ to decide if they sb-deliver a message brb-cast by $\sigma$ or if $\sigma$ is quiet.
In the first case they propose the message they brb-delivered form $\sigma$.
If $\sigma$ is suspected as quiet they abort, i.e. they propose $\bot$ for every sequence number in $S$ for which nothing has been proposed yet.
Note that a correct node aborts only if the \moduleAbbr instance is initialized.
Otherwise, if $\sigma$ is suspected before the intialization event, a correct node could abort permaturely
even though, by the time the instance is initialized, $\sigma$ would have been removed from the list of suspects.
Moreover, upon initialization, all correct nodes should abort if $\sigma$ is already in the list of suspects.
Otherwise, the failure detector may not suspect $\sigma$ again, preventing \moduleAbbr from teminating.

\begin{algorithm}
  \caption{\moduleAbbr implementation with consensus}
  \scriptsize

  \label{alg:sb-impl}
  \begin{algorithmic}[1]
    \Implements
      \State \modulename, \textbf{instance} \textit{\moduleAbbrSmall}, with sender $\sigma$, sequence numbers $S$, messages $M$,
      \State Eventually Strong Failure Detector \textbf{instance} $D$
    \EndImplements
    \Uses
      \State Byzantine Reliable Broadcast, \textbf{instance} \textit{brb}[$sn$] for each $sn\in S$
      \State Byzantine Consensus, \textbf{instance} \textit{bc}[$sn$] for each $sn\in S$
      \State $proposed\leftarrow \emptyset$
      \State $initialized\leftarrow False$
    \EndUses

    \UponEvent{\langle \textit{\moduleAbbrSmall}, \sbinit \rangle}
      \State $initialized\leftarrow True$
      \If{$\sigma\in D.suspected$}
        \label{ln:sb-ealy-suspect}
        \State abort()
      \EndIf
    \EndUponEvent

    \UponEvent{\langle \textit{\moduleAbbrSmall}, \sbcast | sn, m \rangle}
      \label{ln:sb-cast}
      \State \textbf{trigger} $\langle  \textit{brb}[sn], \rbcast | m \rangle$
    \EndUponEvent

    \UponEvent{\langle \textit{brb}[sn], \rbdeliver | m \rangle}
      \label{ln:rb-deliver}
      \State $proposed\leftarrow proposed\cup\{sn\}$
      \State \textbf{trigger} $\langle \textit{bc}[sn], \bcpropose | m \rangle$
    \EndUponEvent

    \UponEvent{\langle D, \suspect | p \rangle}
      \If{$initialized$}
        \label{ln:sb-suspect}
        \State abort()
      \EndIf
    \EndUponEvent

    \UponEvent{\langle \textit{bc}[sn], \bcdecide | v \rangle}
      \State \textbf{trigger} $\langle \textit{\moduleAbbrSmall}, \sbdeliver | sn,v \rangle$
      \label{ln:sb-deliver}
    \EndUponEvent

    \Process{abort}{ }
      \ForAll{sn}{S}
      \label{ln:sb-StartForAllInS}
        \If{$sn \notin proposed$}
          \State $\langle \textit{bc}[sn], \bcpropose | \bot \rangle$
        \EndIf
      \EndForAll
      \label{ln:sb-EndForAllInS}
    \EndProcess
  \end{algorithmic}
\end{algorithm}

\subsubsection{\moduleAbbr implementation correctness}
Below we show that \Cref{alg:sb-impl} satisfies the \moduleAbbr properties.

\noindent \textbf{\moduleProperty1 Integrity: }
If a correct node sb-delivers $(sn, m)$ with $m\neq\bot$ and $\sigma$ is correct then $\sigma$ sb-cast $(sn, m)$.
\begin{proof}
  A correct node sb-delivers $(sn, m)$ only if it has decided on $(sn, m)$ as consensus value (line \ref{ln:sb-deliver}).
  By consensus Validity (\bcProperty4), any value other than $\bot$ is proposed by some correct process.
  Correct nodes invoke consensus for a value $m$ with $m\neq\bot$ only if this value is brb-delivered (line \ref{ln:rb-deliver}).
  By Byzantine reliable broadcast Integrity property (\rbProperty2), since we assume $\sigma$ correct, $sigma$
  brb-casts value $m$.
  This event is only triggered if $\sigma$ sb-casts $(sn, m)$ (line \ref{ln:sb-cast}) which proves \moduleProperty1 Integrity.
\end{proof}

\noindent \textbf{\moduleProperty2 Agreement:}
If two correct nodes sb-deliver, respectively, $(sn, m)$ and $(sn, m')$, then $m = m'$.
\begin{proof}
  If a correct node sb-delivers a value $v$, then $v$ is decided by consensus (line \ref{ln:sb-deliver}).
  By consensus Integrity (\bcProperty3), all correct processes decide  on one value per consensus invocation.
  By consensus Agreement (\bcProperty2), all correct nodes decide on the same value and, therefore, sb-deliver the same value.
\end{proof}

\noindent\textbf{\moduleProperty3 Termination:}
If $p$ is correct, then $p$ eventually sb-delivers a message for every sequence number in $S$,
i.e., $\forall sn \in S : \exists m \in M \cup \{\bot\}$ such that $p$ sb-delivers $(sn, m)$.
\begin{proof}
  Let us assume that there exist some sequence number $sn\in S$ for which a correct node $p$ never sb-delivers any message.
  We distinguish two cases.
  Either some correct node $q$ brb-delivers a message $m$ with $sn$,
  or no correct process ever brb-delivers any message with $sn$.

  In the first case, by BRB Totality (\rbProperty 6), every correct node brb-delivers a message and by BRB Consistency (\rbProperty 5) this message is $m$,
  Therefore, every correct process invokes consensus with $m$  (line \ref{ln:rb-deliver}).
  By consensus Validity (\bcProperty4), every correct process decides on $m$ and, therefore, sb-delivers $(sn,m)$ (line \ref{ln:sb-deliver}).
  A contradiction to correct noded $p$ not sb-delivering any message with $sn$.

  In the second case, by Strong Completeness of the failure detector, $\sigma$ will be eventually suspected to be quiet.
  Therefore, all correct process will abort; they will invoke consensus with $\bot$ for all sequence numbers in $S$, including $sn$ (lines \ref{ln:sb-StartForAllInS}-\ref{ln:sb-EndForAllInS}).
  By consensus Validity (\bcProperty4), every correct process decides on $\bot$ and, therefore, sb-delivers $(sn,\bot)$ (line \ref{ln:sb-deliver}).
  Again, a contradiction to correct noded $p$ not sb-delivering any message with $sn$.
\end{proof}

\noindent\textbf{\moduleProperty4 Eventual Progress:}
If some correct node sb-delivers $(sn,\bot)$ then some correct node suspects $\sigma$ after  $\langle \sbinit \rangle$ is invoked.
\begin{proof}

Let us assume some correct node $p$ sb-delivers $(sn,\bot)$.
By \Cref{alg:sb-impl} either $p$ must have decided on $\bot$ as consensus value (line \ref{ln:sb-deliver}).

By consensus Validity  (\bcProperty4), there exist two possiblies.
Either (1) all correct nodes proposed the same value $\bot$,
or (2) not all correct nodes propose the same value.

In the first case, some correct node $q$ proposed $\bot$.
Therefore, by \Cref{alg:sb-impl}, $q$ suspected $\sigma$ after initialization (lines \ref{ln:sb-ealy-suspect},  \ref{ln:sb-suspect}).

In the second case, by \Cref{alg:sb-impl}, not all correct nodes brb-delivered the same message.
Thus, by reliable broadcast Totality (\rbProperty5), no correct node delivered any message.
Therefore, by reliable broadcast Validity, $\sigma$ is quiet; otherwise $\sigma$ would have brb-cast some message which alld correct nodes would have delivered.
By Strong Completeness of $\diamond S(bz)$, $\sigma$ is eventually permanently suspected by all correct nodes.
\end{proof}

\subsection{From \moduleAbbr to \smr}
\label{sec:smr-proof}
We prove that multiplexing \moduleAbbr instances with \sysname implements an SMR service to a set of clients, as defined in \Cref{sec:model}.
Line numbers refer to \Cref{alg:top-level,alg:auxiliary,alg:epoch-init,alg:leader-selection-policies} in \Cref{sec:details}.

\noindent \textbf{\smrProperty1 Integrity: }
If a correct node delivers $(sn, r)$,
where $r.id.c$ is a correct client's identidy,
then client $c$ broadcast $r$.

\begin{proof}
A correct node only delivers request $r$ if it is inserted in the log as part of a committed batch $b$ (line \ref{ln:in-order-delivery}).
In turn, $b$ is added in the log only upon an event $\langle sb, \sbdeliver | sn_b,b \rangle$, where $sb$ is a $\moduleAbbr(\sigma,M,S,D)$ instance (line \ref{ln:module-deliver}) and
$M$ the set of all possible valid batches in the buckets of the segment.
A correct node only invokes $\langle sb, \sbdeliver | sn_b,b \rangle$ with a batch $b$ in the set of valid batches $M$.
A validity condition is that every request of the batch has a valid signature (see~\Cref{sec:requests}).
Since $r.id.c$ is the only process able to produce a valid signature, $c$ must have broadcast $r$.
\end{proof}


\noindent \textbf{\smrProperty2 Agreement:}
If two correct nodes deliver, respectively, $(sn,r)$ and $(sn,r')$, then $r=r'$.

\begin{proof}
Let $r$ be in batch $b$ committed with $sn_{b}$ and $r'$ in a batch $b'$ committed with $sn_{b'}$.
For $r$, $r'$ to have the same sequence number, it follows from \Cref{eq:requests} and by the same log established by \moduleProperty2 (Agreement) that $sn_b=sn_b'$.
Since $b$ and $b'$ are delivered with the same sequence number, they belong to the same segment $S$ and, thus, also in the same set $M$ for which an instance $\moduleAbbr(\sigma,M,S,D)$ was initialized.
Thus, by \moduleProperty2 $b=b'$,
and by \Cref{eq:requests} $r=r'$.
\end{proof}



\begin{lemma}\label{lemma:epoch-liveness}
If a correct node initializes $\moduleAbbr(\sigma,M,S,D)$ then every correct node eventually initializes $\moduleAbbr(\sigma,M,S,D)$.
\end{lemma}
\begin{proof}
By the \sysname algorithm, the \moduleAbbr instances (including their parameters $\sigma, M, S,D$) are initialized at the beginning of each epoch $e$ depending only on the state of the node's log at the beginning of $e$.
We prove the lemma by induction on the epoch $e$.
In epoch $0$ all nodes have an empty log.
Thus, they instantiate $SB$ with the same parameters.

We assume that in epoch $e$ correct nodes $i$ and $j$ initialize \moduleAbbr with the same parameters.
Then, by \moduleAbbr~3 (Termination) and by \smrProperty~2 (Agreement) both $i$ and $j$ advance from epoch $e$ to epoch $e+1$ with the same log.
Thus, both $i$ and $j$ initialize \moduleAbbr with the same parameters, i.e., sender $\sigma$, message set $M$ and segments $S$.
\end{proof}

\noindent \textbf{\smrProperty3 Totality:}
If a correct node delivers request $(sn,r)$, then every correct node eventually delivers $(sn,r)$.
\begin{proof}
Let us assume that some correct node $i$ delivers $r$ with sequence number $sn$.
Let $j$ be some other correct node.
From~\Cref{eq:requests} request $r$, delivered by $i$, uniquely corresponds to some batch $b$ with sequence number $sn_b$ in the log of $i$ which $i$ has committed.
Therefore, there exists an instance $\moduleAbbr(\sigma,M,S,D)$ which outputs $\sbdeliver(sn_b,b)$ at node $i$.
By~\Cref{lemma:epoch-liveness}, node $j$ also eventually initializes $\moduleAbbr(\sigma,M,S,D)$.
Then \moduleAbbr~3 (Termination) guarantees that for each sequence number in $S$ and, therefore, for also $sn_b$, $j$ delivers a message $m\in M\cup\bot$.
Moreover, property \moduleAbbr~2 (Agreement) guarantees that $m=b$.
It follows that $j$ delivers $r$ for $sn$.
\end{proof}

\begin{lemma}\label{lemma:unsuspected-correct-leader}
In an infinite \sysname execution, there exists some correct node which is eventually in the leaderset forever without being suspected.
\end{lemma}

\begin{proof}
We prove the lemma for all three suggested leader policies.
By the Eventual Weak Accuracy of the failure detector, there exists some correct node $p$ that will eventually stop being suspected by any correct node.
  \begin{itemize}
    \item \textbf{SIMPLE:} All nodes are in the leaderset in every epoch, therefore also $p$ is in the leaderset in every epoch.
    \item \textbf{BACKOFF:} The policy guarantees that any node which has been removed from the leaderset will be re-included.
    Since node $p$ will stop being suspected, eventually $p$ stays in the leaderset forever.
    \item \textbf{BLACKLIST:} Eventually, either $p$ will be in the leaderset forever, or there exists at least $2f+1$ other nodes that stop being suspected forever -- otherwise $p$ would eventually be included in the leaderset.
    Among those $2f+1$ nodes there are at least $f+1$ correct.
    Therefore, in any case, some correct node which stops being suspected remains in the leaderset forever.
  \end{itemize}
\end{proof}

\begin{lemma}\label{lemma:sb-validity}
Let $p$ be a correct node which is eventually in the leaderset forever without being suspected by any correct node after time $t$.
If $p$ sb-casts a message $m$ after time $t$, then $p$ eventually sb-delivers $m$.
\end{lemma}
\begin{proof}
By \moduleAbbr~3 (Termination) $p$ sb-delivers a message for all sequence numbers for any \moduleAbbr instance.
Let us assume that $p$ sb-casts message $m$ after time $t$ in a \moduleAbbr instance but delivers a set of messages $\mathcal{M}$ for the sequence numbers of the instance, such that $m\notin \mathcal{M}$.
By $p$ not being suspected and \moduleAbbr~4 (Eventual Progress), for the \moduleAbbr instances for which it is the dedicated sender, $p$ delivers non $\bot$ values for all sequence numbers.
By \moduleAbbr~1 (Integrity) all messages in $M$ are sb-cast by $p$.
A contradiction to $m\notin \mathcal{M}$, since $p$ being correct, sb-casts only one message per sequence number.
Therefore, $p$ sb-delivers $m$.
\end{proof}

\begin{lemma}\label{lemma:infinite-prposals}
Let $p$ be a correct node that after time $t$ is in the leaderset forever and not suspected by any correct node.
$p$ is assigned every bucket infinitely many times.
\end{lemma}

\begin{proof}
Since $p$ is the leaderset in every epoch after time $t$, there exists an epoch number $e$ such that, for $i=p$, satisfies \cref{eq:buckets} for every bucket $\hat{b}$ in the set of buckets $\mathcal{B}$.
\end{proof}

\begin{lemma}\label{lemma:oldest-delivered}
Let $p$ be a correct node that after time $t$ is in the leaderset forever and not suspected by any correct node.
Let $r$ be the oldest request that $p$ has received.
Let $\hat{b}$ be bucket $r$ maps to.
Some correct node eventually delivers $r$.
\end{lemma}

\begin{proof}
  Let $e_t$ be the epoch number of the first epoch which starts at all nodes after time $t$.
  There exist two mutually exclusive cases.
  \begin{enumerate}
    \item Some correct node delivers $r$ in an epoch before epoch $e_t$.
    \item No correct node delivers $r$ in an epoch before epoch $e_t$.
  \end{enumerate}
  In the first case the lemma is trivially satisfied.
  We will prove the second case by contradiction.

  Let us assume that no correct node delivers $r$ in an epoch with number $>=e_t$.
  By the \sysname algorithm, $p$ only ever removes $r$ from its bucket queue if $p$ proposes (line \ref{ln:remove-proposed-req}) or delivers (line \ref{ln:remove-delivered-req}) $r$.
  Since, by our assumption, $p$ does not deliver $r$ before epoch $e_t$, $r$ remains in $p$'s corresponding bucket queue.
  Even if $p$ proposed $r$ before epoch $e_t$, since $r$ is not delivered, by the \sysname algorithm $p$ has resurrected $r$ (line \ref{ln:resurrect-call}), i.e. $r$ is re-added in the corresponding bucket-queue maintaining its reception order.
  By \Cref{lemma:infinite-prposals}, $p$ is eventually assigned bucket $\hat{b}$.
  Since $r$ is the oldest request, by the \sysname algorithm, $p$ will sb-cast a batch $b$ containing $r$ (lines \ref{ln:propose-cut-batch},\ref{ln:oldest}).
  By \Cref{lemma:sb-validity}, $p$ sb-delivers $b$.
  Let $sn$ be the sequence number with which $b$ is sb-delivered.
  By \moduleProperty3 (Termination), all correct nodes sb-deliver and add in their log all seqence numbers in the segment of $sn$.
  Therefore, by the \sysname algorithm (line \ref{ln:in-order-delivery}), they deliver all requests in $b$, including $r$.
  A contradiction to $r$ not being delivered.
\end{proof}

\noindent \textbf{\smrProperty4 Liveness:}
If a correct client broadcasts request $r$, then some correct node eventually delivers $(sn, r)$.
\begin{proof}
Let us assume by contradiction that $r$ is never delivered by any correct node.
This implies that every correct node puts $r$ in their respective bucket queue (by the correct client re-transmitting $r$ forever, see \Cref{sec:requests}).
Eventually, by \Cref{lemma:unsuspected-correct-leader} there will be at least one correct, unsuspected node $i$ in the leaderset forever.
Let $r$ map to a bucket $\hat{b}$.
Let $R(r)$ be the set of all requests received by $i$ before receiving $r$.
We prove Liveness by induction on the size of $R(r)$.
For $|R(r)|=0$, $r$ is the oldest request that $i$ has received and, by \cref{lemma:oldest-delivered}, $r$ is delivered by some correct node.
For the induction step we show that if $r$ is delivered for $|R(r)|<=d$, then $r$ is delivered for $|R(r)|=d+1$.
Let $|R(r)|=d+1$.
We denote by $r^k$ be the request $k^{th}$ oldest request, such that $r^{d}$ is the request received immediately before $r$ with $R(r^d)|=d$.
By the induction hypothesis, all $r^k$ for $k<=d$ will be delivered.
Therefore, $r$ becomes the oldest request and, by \cref{lemma:oldest-delivered}, $r$ is delivered by some correct node.
A contradiction to $r$ not being delivered.
\end{proof}

\section{Evaluation}
\label{sec:evaluation}

Our implementation is modular, allowing to easy switching between different protocols implementing \moduleAbbr.
We use 3  well-known protocols for ordering requests: PBFT~\cite{Castro:2002:PBF} (BFT), HotStuff~\cite{hotstuff} (BFT), and Raft~\cite{RAFT} (CFT).
We evaluate the impact \sysname has on these protocols by comparing its performance to their respective original single-leader versions.
In addition, we compare \sysname to Mir-BFT~\cite{stathakopoulou2019mir} which also has multiple leaders.
We do not compare, however, to other multi-leader protocols that do not prevent request duplication
(e.g., Hashgraph \cite{Hashgraph}, Red Belly \cite{RedBelly}, RCC \cite{rcc}, OMADA \cite{omada}, BFT-Mencius~\cite{BFT-Mencius}).
The codebase of these protocols is unavailable or unmaintained.
Moreover, Mir-BFT evaluation demonstrates that the performance of this family of protocols deteriorates as the number of
nodes increases in the presence of duplicate requests.
For the same reason, we also don't compare to trivially running multiple instances of the single leader protocols.
We first evaluate performance in absence of failures.
Then, we study how \sysname behaves when failures occur.

\subsection{Experimental Setup}
We perform our evaluation in a WAN which spans in 16 datacenters across Europe, America, Australia, and Asia on IBM cloud.
All processes run on  dedicated virtual machines with 32 x 2.0 GHz VCPUs and 32GB RAM running Ubuntu Linux 20.04.
All machines are equipped with two network interfaces, public and private, rate limited for repeatability to 1 Gbps; the public one is for request submission and the private one is for node-to-node communication.
Clients submit requests with 500 byte payload,
the average Bitcoin transaction size \cite{BitcoinTXSizeURL}.
Each node runs on a single virtual machine. Each node setup is uniformly distributed across all datacenters, except for the $4$-node setup which spans on $4$ detacenters, distributed across all $4$ continents.
We use 16 client machines, also uniformly distributed across all datacenters, each running 16 clients in parallel which communicate independently with the nodes using TLS.
We evaluate throughput, i.e., the number of requests the system delivers per second, and end-to-end latency, i.e., the latency from the moment a client submits a request until the client receives $f+1$ responses.

\subsection{\sysname Configuration}
\label{sec:configuration}
After a preliminary evaluation of each protocol, we concluded to a meaningful set of
parameters.
We do not claim that this set is optimal; an exhaustive set of experiments to find
the optimal configuration goes beyond the scope of this work.
Our choice of parameters allows us, however, to demonstrate that
\sysname makes all three protocols scalable, which is the key contribution of this work.
We here on discuss the most critical configuration parameters.
\Cref{table:config} summarizes the set of parameters of our evaluation.

For Raft and PBFT we maintain a fixed batch rate.
This translates to $O(1/n)$ proposals per leader and $O(n)$ message complexity per bottleneck
node\footnote{We use the term \emph{bottleneck} node to refer to the node that processes the most messages in
each protocol.}.
The choice of a fixed batch rate prevents throughput from dropping due to super-linear message complexity.
On the other hand, it introduces higher end-to-end latency as the number of nodes grows,
since the batch timeout increases.

The epoch length is kept short:
$256$ bathces per epoch for a batch rate of $32$ batches per second yield an epoch duration of
aproximately $8$ seconds in a fault free execution%
\footnote{The actual epoch duration is increased by the time a batch needs to commit and by the overhead of extra communication rounds at the end of the epoch.
Epoch duration becomes shorter if the load is high such that batches become full before the batch timeout.}.
Shorter epoch length maintains lower latencies in case a fault occurs because bucket re-distribution
is executed faster.
However, a fixed epoch length yields a shorter segment length as the number of nodes increases.
Too short segments for HotStuff and Raft translate to a significant
overhead of the dummy/empty bathces at the end of the segment.
We, therefore, chose a larger minimum segment size for those two protocols.

Batch timeouts should in general be kept small to prevent increasing end-to-end latency.
In Raft, however, by design, a leader node re-sends proposals
until it has received an acknowledgement from the followers.
A very short batch timeout would result in sending proposals too soon and therefore repeating previous proposals.
This has a negative impact on throughput because the bandwidth is consumed by unnecessary duplicate proposals.
To avoid this phenomenon, we opted for a minimum batch timeout longer than the approximated network round trip.
To prevent rate-limiting Raft throughput due to the long batch timeouts, we allowed a large batch size.

HotStuff, on the other hand, is a latency bound protocol.
This is because sending a proposal first requires assembling a quorum certificate which depends on the previous proposal.
We opted, therefore, for a batch timeout of $0$ to allow the leader to send proposals as fast as possible.
Similarly to Raft, we allowed a large batch size to prevent rate-limiting the throughput.

\begin{table}[ht]
\centering
\scriptsize{
\begin{tabular}{| l | l | l | l | }
\hline
& \textbf{PBFT} & \textbf{HotStuff} & \textbf{Raft} \\ \hline
Initial lederset size & $|\mathcal{N}|$ &  $|\mathcal{N}|$ & $|\mathcal{N}|$ \\ \hline
Max batch size & 2048 & 4096 & 4096 \\ \hline
Batch rate & 32 b/s & not applicable & 32 b/s\\ \hline
Min batch timeout & 0 s & 1 s & 0 s \\ \hline
Max batch timeout & 4 s & 0 & 4 s \\ \hline
Min epoch length & 256 & 256 & 256 \\ \hline
Min segment size & 2 & 16 & 16 \\ \hline
Epoch change timeout & 10 s & 10 s &  [10,20) s \\ \hline
Buckets per leader & 16 & 16 & 16 \\ \hline
Client signatures & 256-bit ECDSA & 256-bit ECDSA  & none\\ \hline
\end{tabular}
}
\caption{\sysname configuration parameters used in evaluation}
\label{table:config}
\end{table}

\subsection{Failure-Free Performance}

\Cref{fig:scalability} shows the overall throughput scalability of PBFT, HotStuff, and Raft, with and without \sysname,
as well as that of Mir-BFT\@.

We evaluate the scalability of \sysname with up to 128 nodes, uniformly distributed across all 16 datacenters.
Mir-BFT is evaluated on the same set of datacenters on machines with the same specifications.
For a meaningful, apples-to-apples, comparison, we disabled Mir-BFT optimizations (signature verification sharding and light total order broadcast).
Such optimizations could be implemented on top of \sysname yielding even better performance.
However, this goes beyond the scope of this work.
For all protocols we run experiments with increasing the client request submission rate until the throughput is saturated.

In \Cref{fig:scalability} we report the highest measured throughput before saturation.
We observe that \sysname dramatically improves the performance of the single leader protocols as the number of nodes grows
(37x, 56x and 55x improvement for PBFT, HotStuff, and Raft, respectively, on 128 nodes).
This improvement is due to overcoming the single leader bandwidth bottleneck.
Moreover, as the number of nodes grows, \sysname-PBFT outperforms Mir-BFT.
While in theory, in a fault-free execution we would expect the two protocols to perform the same, we attribute this improvement to the more careful concurrency handling in the \sysname implementation.

\sysname-PBFT maintains more than $58$ kreq/s on $128$ nodes.
Its performance, though, drops compared to smaller configurations.
We attribute this to the increasing number of messages each node processes,
which, with a fixed batch rate (\Cref{table:config}), grow linearly with the number of nodes.
We further observe that throughput increases for Raft and HotStuff \sysname implementations with the number of nodes, approaching that of \sysname-PBFT.
While PBFT’s watermarking mechanism allows the leader to propose batches in parallel,
HotStuff,
as explained in \Cref{sec:configuration},
is latency-bound.
However, running multiple independent protocol instances with \sysname helps improve the overall throughput.
Raft, on the other hand, suffers from the redundant re-proposals.
While this drawback is mostly hidden in fast LANs with negligible latency, it manifests strongly in a WAN.
With more nodes in the leaderset the batch timeout increases and re-proposals are reduced.

\begin{figure}[h!]
\centering
\includegraphics[width=0.9\columnwidth]{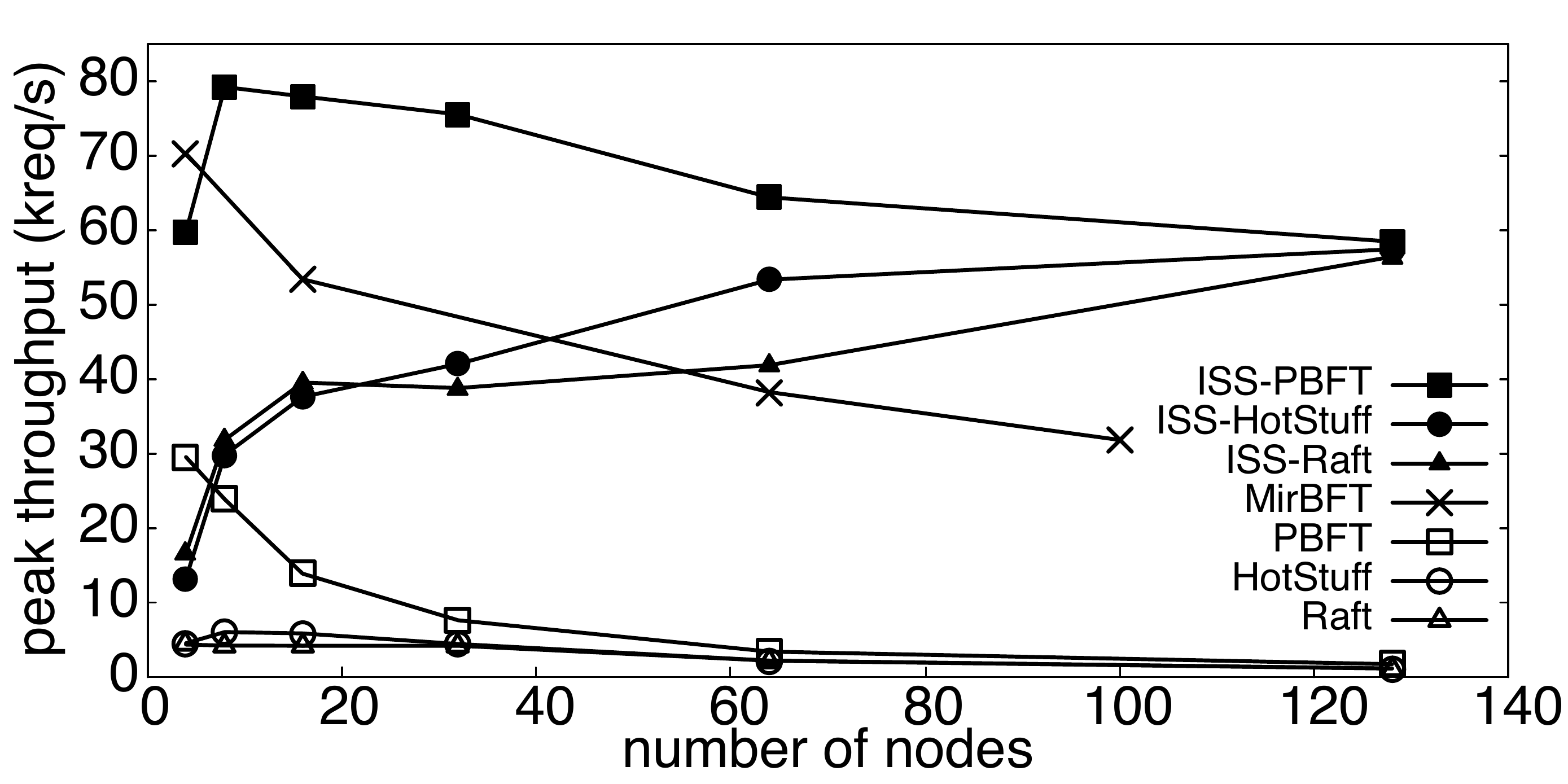}
\caption{Scalability of single leader protocols, their \sysname counterpart, and MirBFT.}
\label{fig:scalability}
\vspace*{-5mm}
\end{figure}

In \Cref{fig:lt} we observe that \sysname latency grows with the number of nodes.
This is due to our choice of a fixed batch rate in order to reduce message complexity and sustain high throughput
with an increasing number of nodes.

\begin{figure}[h!]
\centering
\begin{subfigure}{\columnwidth}
\centering
\includegraphics[width=0.9\columnwidth]{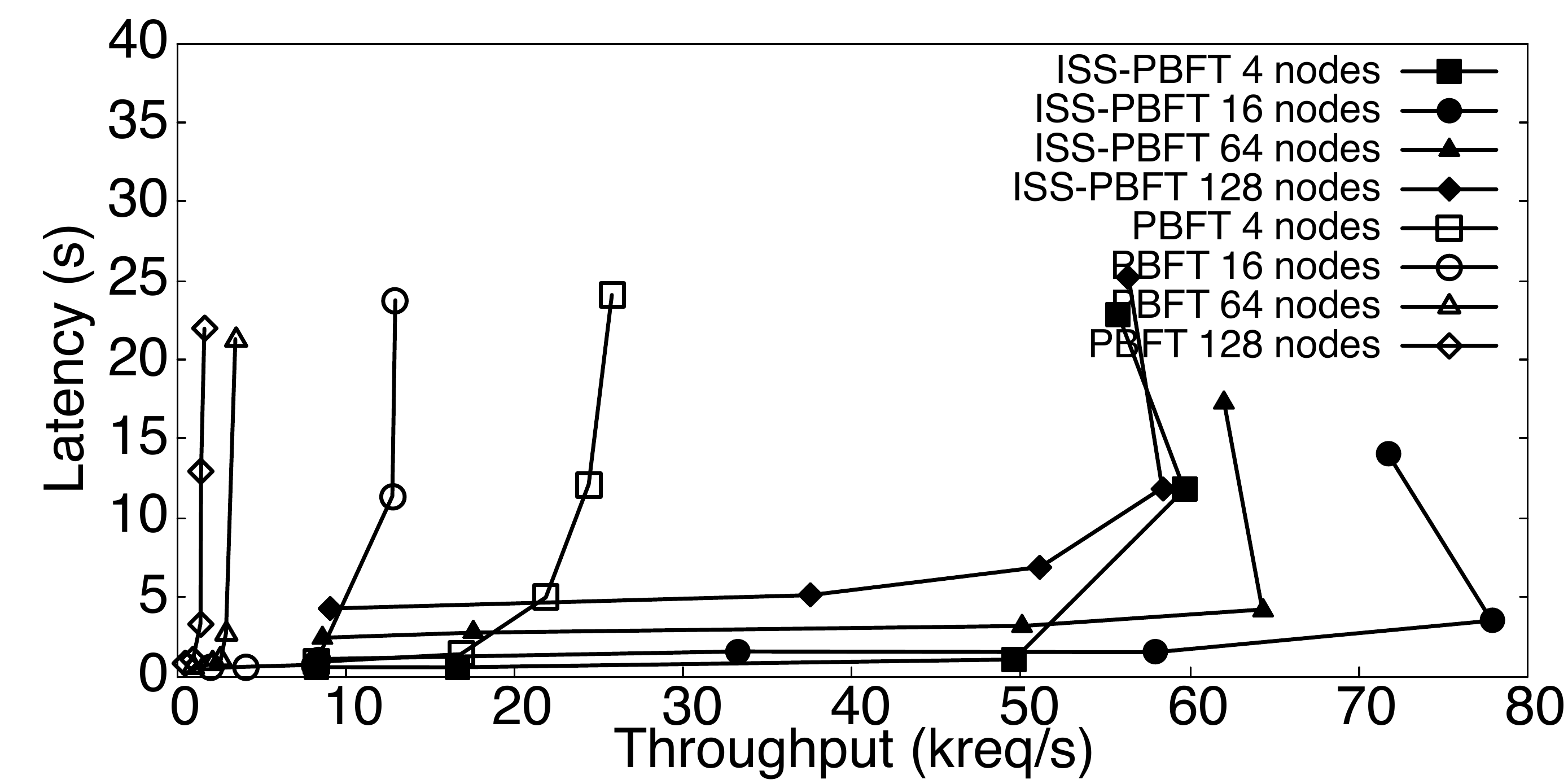}
\label{fig:lt-iss-pbft}
\end{subfigure}
\begin{subfigure}{\columnwidth}
\centering
\includegraphics[width=0.9\columnwidth]{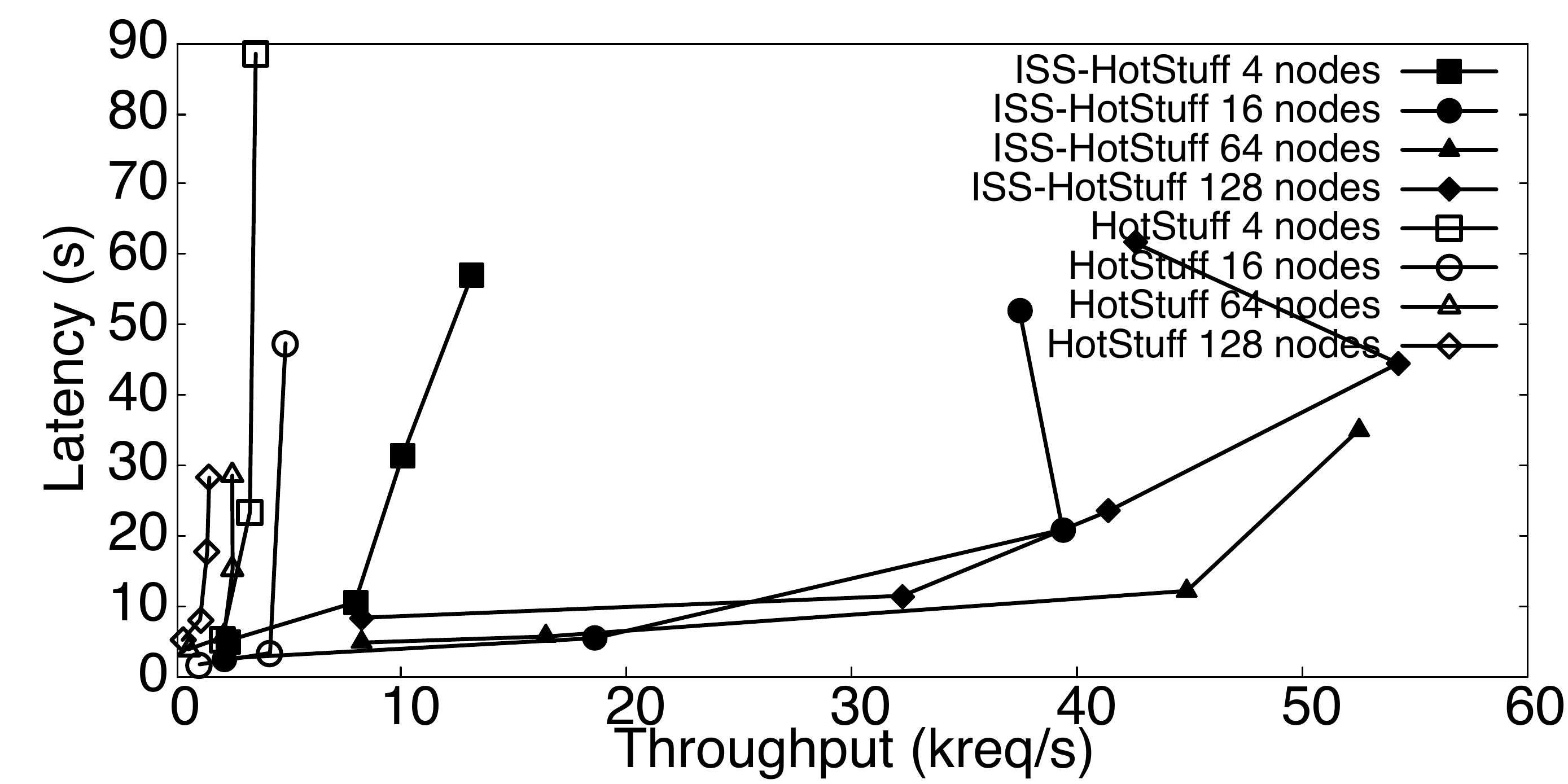}
\label{fig:lt-iss-hotstuff}
\end{subfigure}
\begin{subfigure}{\columnwidth}
\centering
\includegraphics[width=0.9\columnwidth]{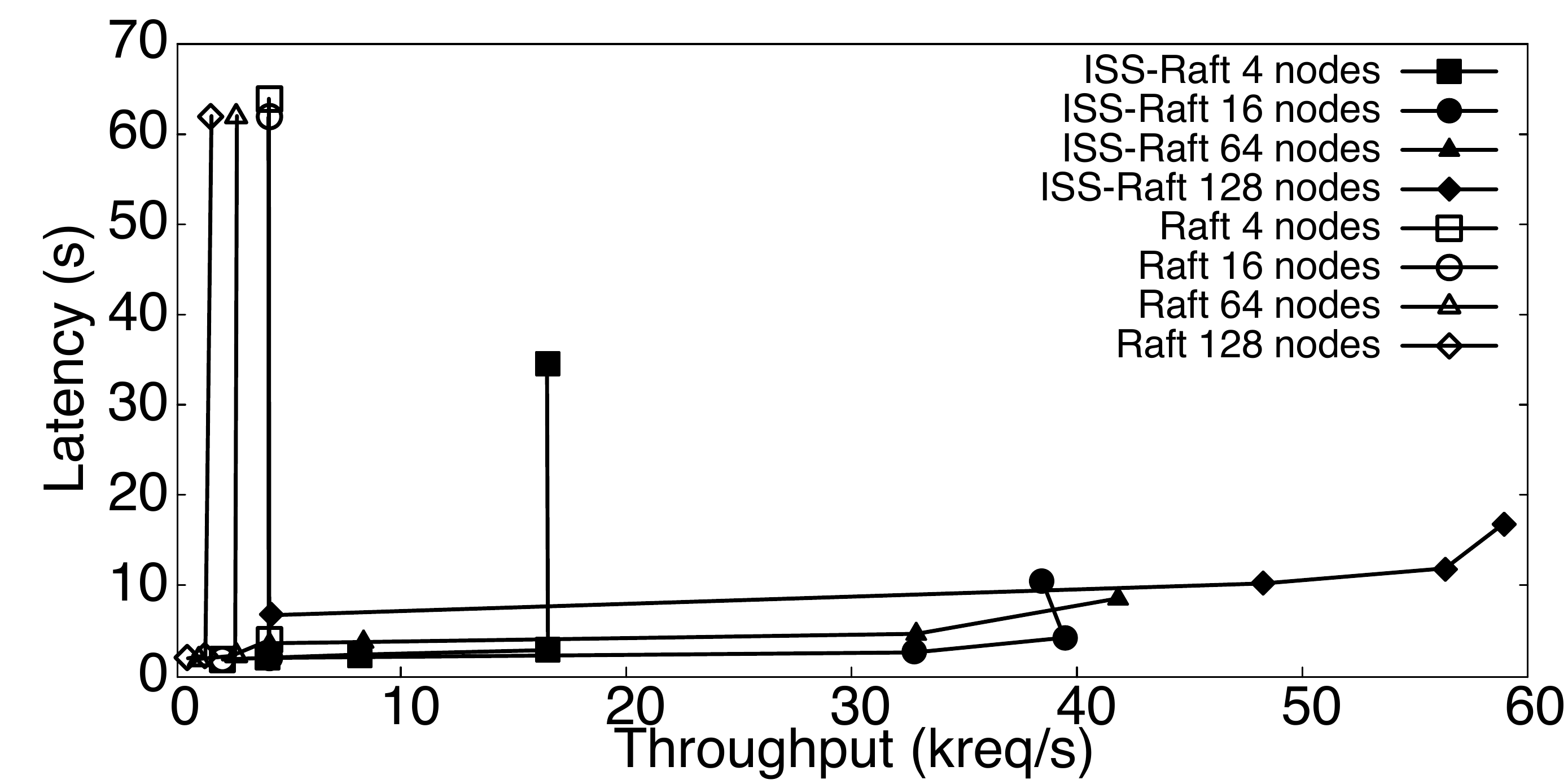}
\label{fig:lt-iss-raft}
\end{subfigure}
\caption{Latency over throughput for increasing load for (a) \sysname-PBFT, (b) \sysname-HotStuff, and (c) \sysname-Raft.}
\label{fig:lt}
\end{figure}

\subsection{\sysname Under Faults}
In this section we fix PBFT as the protocol multiplexed with \sysname and study its performance under crash faults
and Byzantine stragglers in a WAN of $32$ nodes.
The PBFT view change timeout is set at $10$ seconds.

\subsubsection{Crash Faults}
In this section we study how crash faults affect the \sysname latency and throughput
and the impact of the different leader selection policies.

We study two edge cases of faults: (a) one or more leaders crash at
the beginning of the first epoch of the execution and
(b) one or more leaders crash before sending the proposal for the last sequence
number they led in the first epoch of the execution.
The \emph{epoch-start} crash fault is a worst-case scenario for the number of
proposed sequence numbers in an epoch.
The \emph{epoch-end} crash fault is a worst-case scenario for the duration of the epoch;
all nodes need to wait for the fault to be detected at the end of the epoch.

In \Cref{fig:policies} we study the impact of epoch-start and epoch-end crash faults
on end-to-end latency with different leader selection policies for one crash fault
for a two-minute execution which yields 13 epochs.
All experiments are conducted with a total incoming request rate from the clients of
16.4k req/s.
Overall, we observe that \sysname-BFT sustains mean end-to-end latency below 8 seconds
and tail latency below 17 seconds with any leader selection policy.
\begin{figure}[h!]
\centering
\includegraphics[width=0.8\columnwidth]{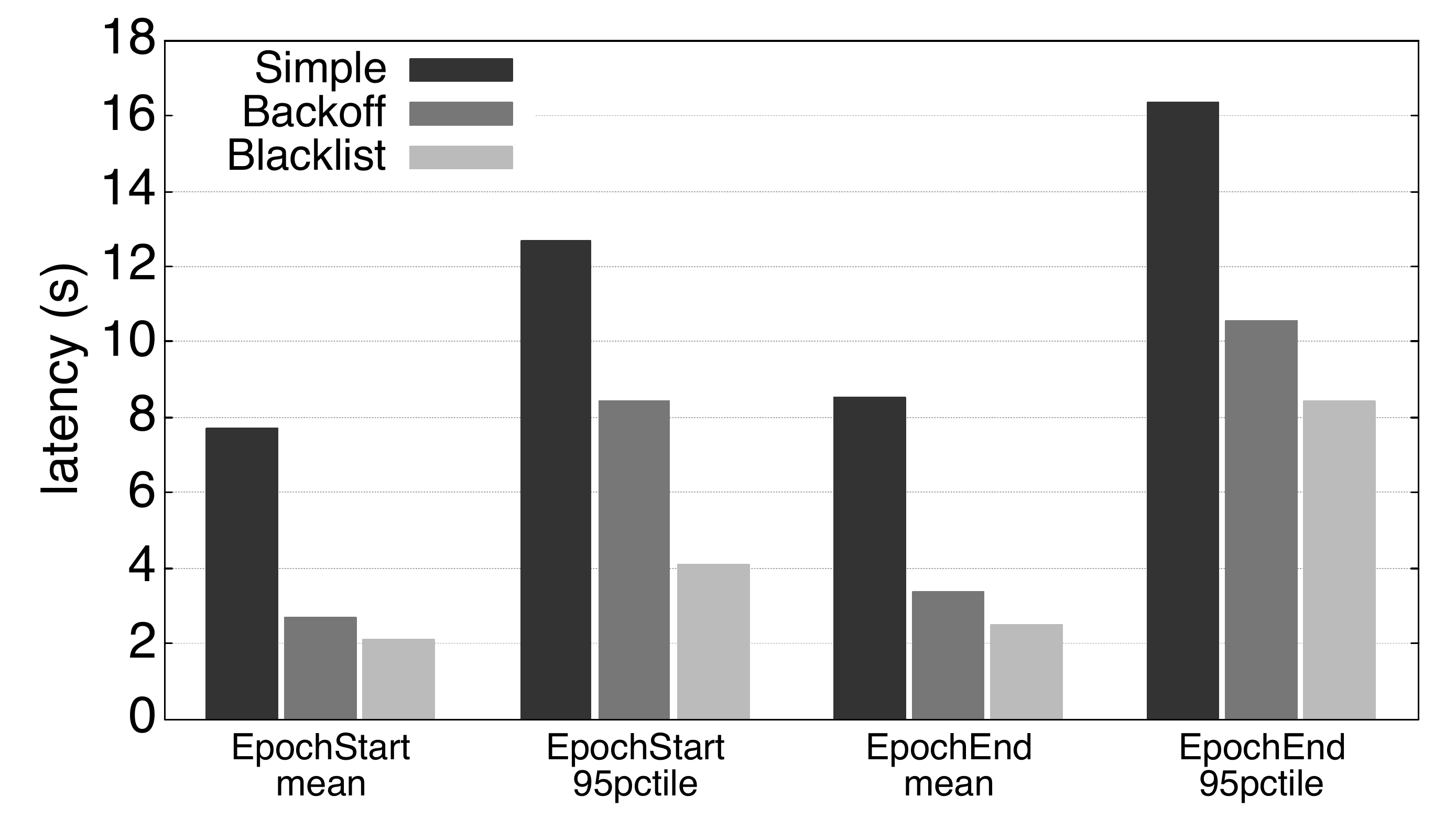}
\caption{Impact of different leader selection policies on end-to-end mean and tail (95 percentile).}
\label{fig:policies}
\end{figure}

Comparing the different leader selection policies,
we can see that Blacklist and Backoff policies maintain lower latency, since those
policies remove the crashed node from the leaderset.
In particular, Blacklist policy performs best by permanently removing the crashed node.
Therefore, latency is only affected during the first epoch.
Based on this conclusion, for the rest of the evaluation under faults we stick with the Blacklist leader selection policy.

\Cref{fig:crash-faults-duration}
shows the impact of crash faults on latency.
We see that latency converges towards that of a fault-free execution as we
increase the duration of the experiment.
This is due to Blacklist leader selection policy removing the faulty node from the leaderset once detected.
Note that, regardless of their number,
epoch-end failures have a stronger impact on latency (as they delay requests in all bucket queues)
than epoch-start failures (affecting only the faulty nodes' bucket queues).

\begin{figure}[h!]
\centering
\begin{subfigure}{\columnwidth}
\centering
\includegraphics[width=0.9\columnwidth]{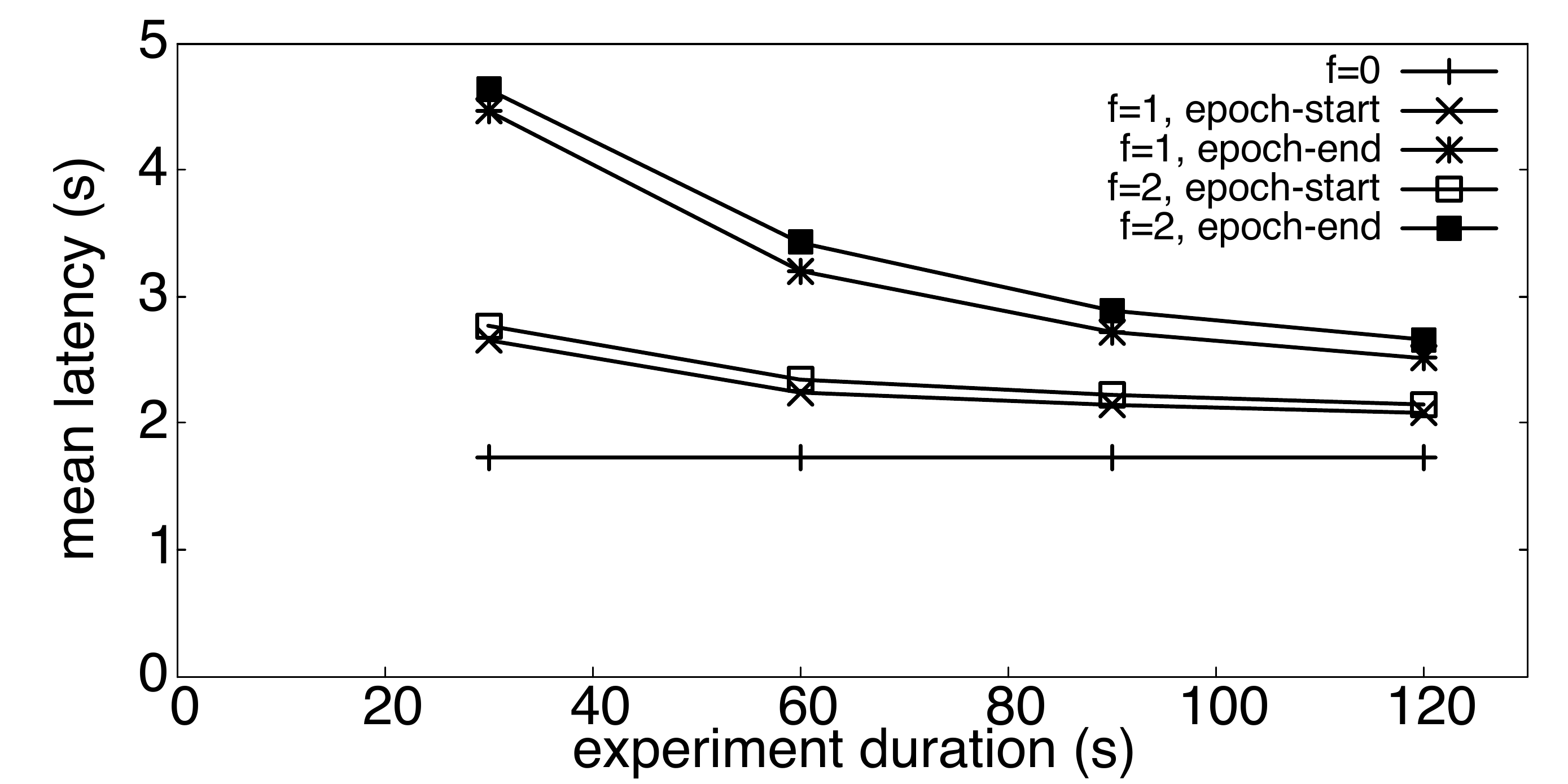}
\label{fig:crash-faults-duration-avg}
\end{subfigure}
\begin{subfigure}{\columnwidth}
\centering
\includegraphics[width=0.9\columnwidth]{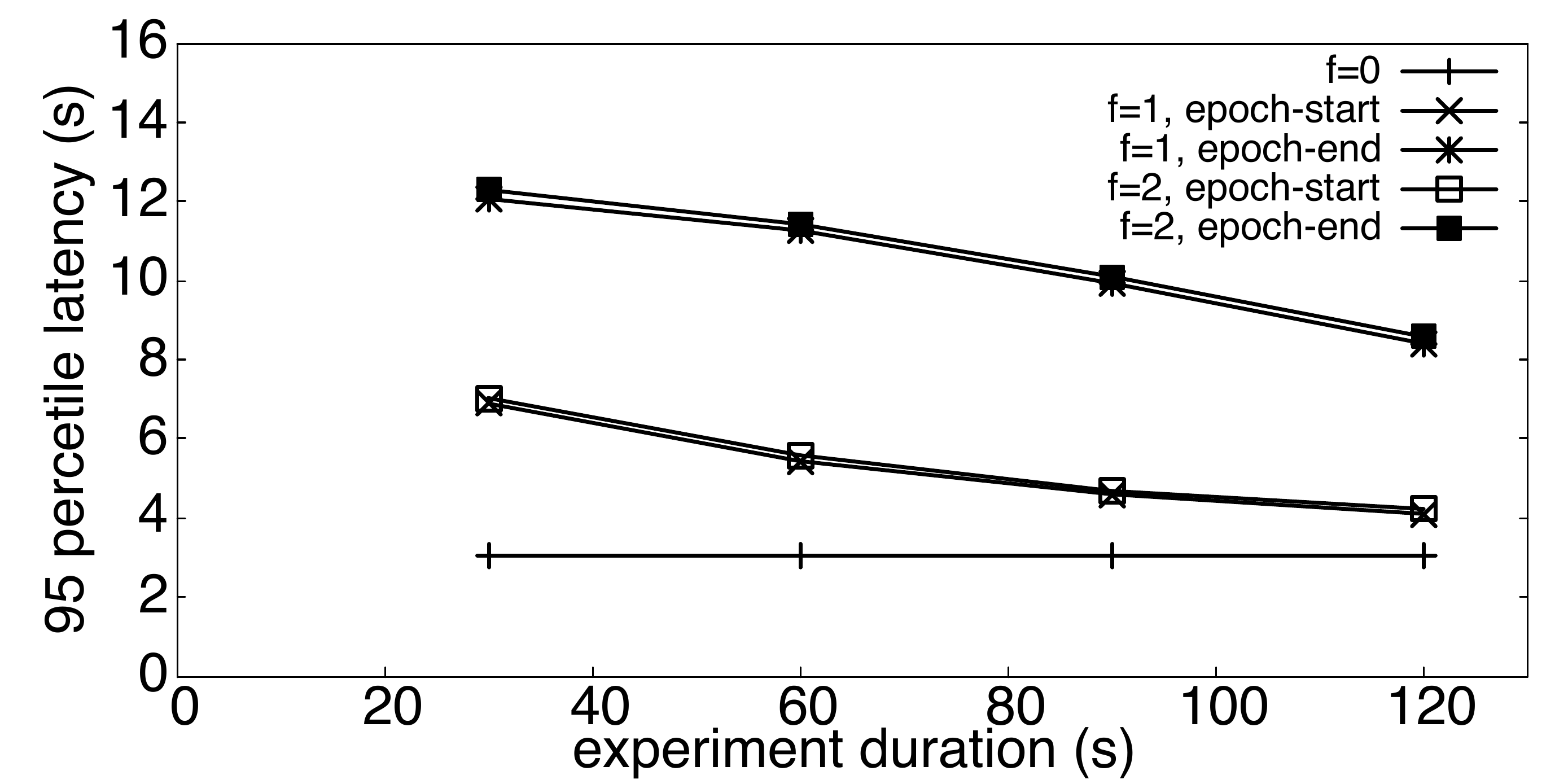}
\label{fig:crash-faults-duration-95p}
\end{subfigure}
\caption{Impact of crash faults on mean (a) and tail (b) end-to-end latency for increasing experiment duration with Blacklist policy.}
\label{fig:crash-faults-duration}
\end{figure}

\Cref{fig:crash-faults-timeline} shows throughput over time.
The short drops to $0$ in throughput correspond to the epoch change.
We see that an epoch-start fault does not delay the epoch change,
as the timer detecting the faulty leader of one segment runs in parallel
with other leaders agreeing on requests in their segments.
On the other hand, the epoch-end fault delays the epoch change.
However, \sysname quickly recovers by ordering more than $170$k req/s at the beginning of the second epoch (see the spike in \Cref{fig:crash-faults-timeline}(b)).

\begin{figure}[h!]
\centering
\begin{subfigure}{\columnwidth}
\centering
\includegraphics[width=0.8\columnwidth]{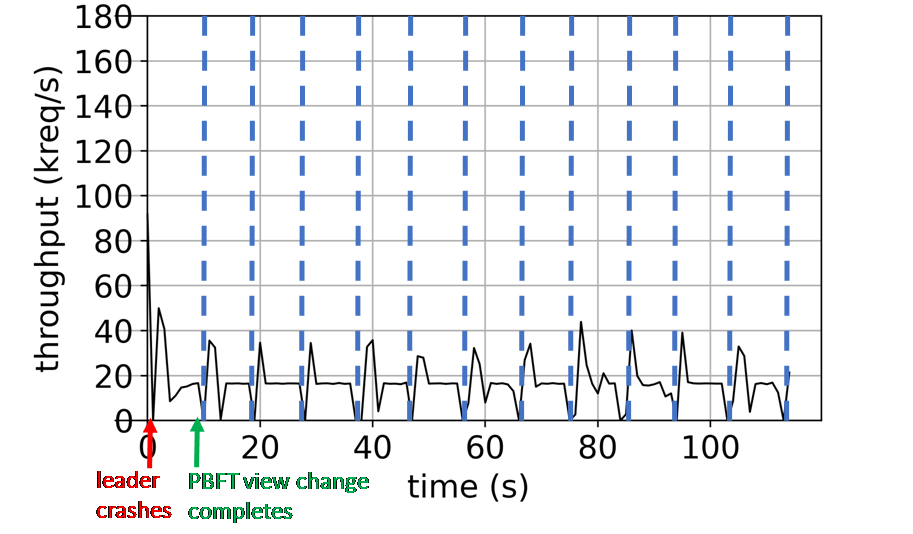}
\label{fig:epoch-start}
\end{subfigure}
\begin{subfigure}{\columnwidth}
\centering
\includegraphics[width=0.8\columnwidth]{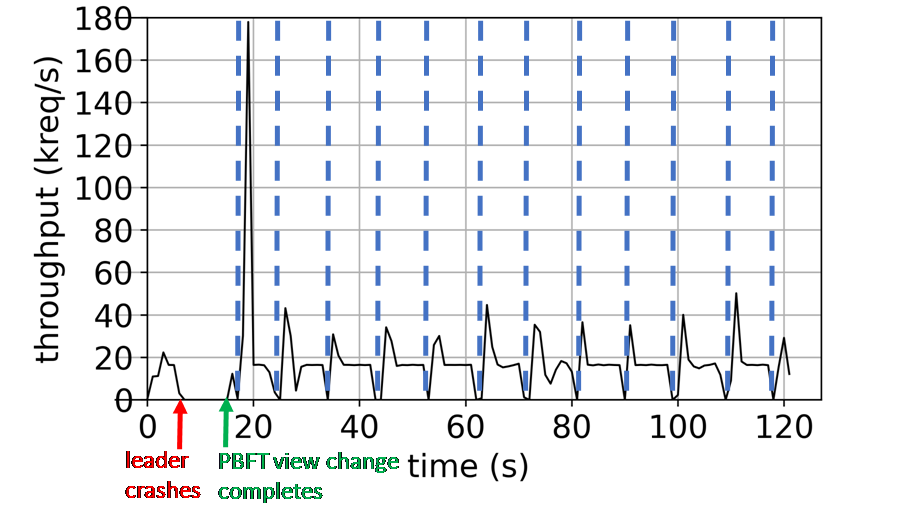}
\label{fig:epoch-end}
\end{subfigure}
\caption{\sysname-PBFT throughput average (over 1s intervals) over time with one crash fault at the beginning (a) and at the end (b) of  the first epoch with Blacklist policy. The dashed lines indicate the end of an epoch.}
\label{fig:crash-faults-timeline}
\end{figure}

We compare the \sysname performance under crash faults to MirBFT.
In \Cref{fig:mir-crash-faults-timeline-epoch-start} we study run MirBFT on $32$ nodes with a single epoch-start crash fault.
MirBFT stops processing any message during the epoch changes, unlike \sysname where segments make progress
independently.
This results in any crash fault having an impact similar to that of the epoch-end fault for \sysname.
Moreover, MirBFT relies on an epoch primary for liveness.
Every time the crashed node becomes epoch primary it causes an ungraceful epoch change.
In \Cref{fig:mir-crash-faults-timeline-epoch-start} this happens around $t=600$.
The phenomenon repeats periodically, unlike \sysname, where the faulty leader is permanently removed.
Finally, \sysname crash fault recovery is more lightweight, since it concerns only
the batches of a single segment.

\begin{figure}[h!]
\centering
\includegraphics[width=0.95\columnwidth]{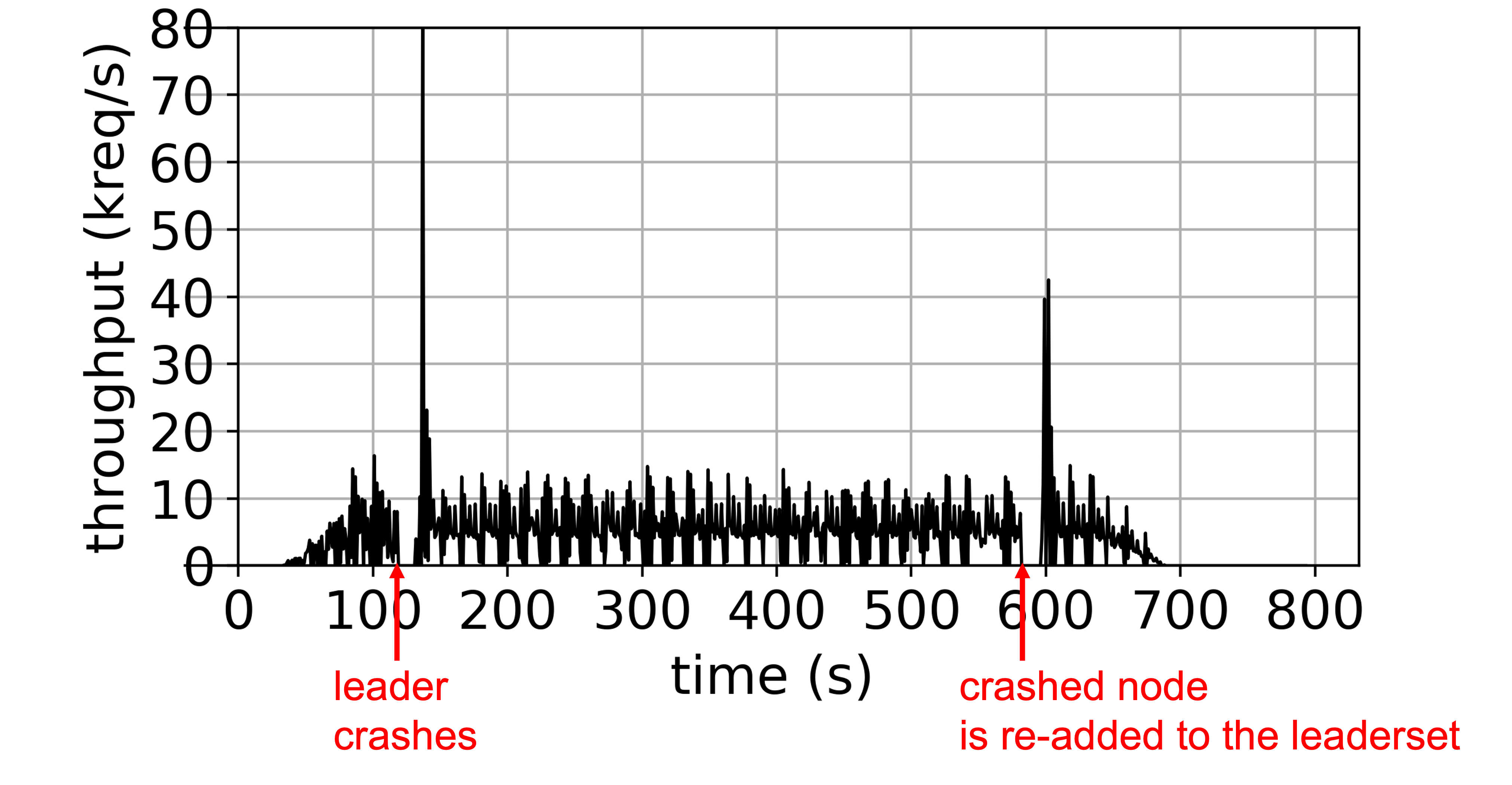}
\caption{MirBFT throughput average (over 1s intervals) over time with one epoch-start fault. Epoch change timeout is at 10s and epoch duration is 256 blocks. Periods of $0$ throughput repeat periodically.}
\label{fig:mir-crash-faults-timeline-epoch-start}
\end{figure}

\subsubsection{Byzantine Stragglers}
In this section we study the impact on latency and throughput of a Byzantine straggler.
A Byzantine straggler delays proposals as much as possible without being suspected as
faulty and does not add requests in its proposals to harm latency and throughput.
We evaluate latency and throughput with $f=1$ up to the maximum tolerated number of $f=10$ stragglers.
In our evaluation the straggler sends out an empty proposal every 0.5x epoch change timeout ($5$ seconds).

\Cref{fig:lt-stragglers} shows the impact of an increasing number of stragglers.
\sysname with PBFT reaches from 15\% of its maximum throughput with one straggler
to 10\% of its maximum throughput with 10 stragglers.
This, though, translates to maintaining more than $11.4$ and $7.9$ kreq/s, respectively, on $32$ nodes.
Mean latency before saturation increases from $14x$ with one up to $29x$ with 10 stragglers.

\begin{figure}[h!]
\centering
\includegraphics[width=0.9\columnwidth]{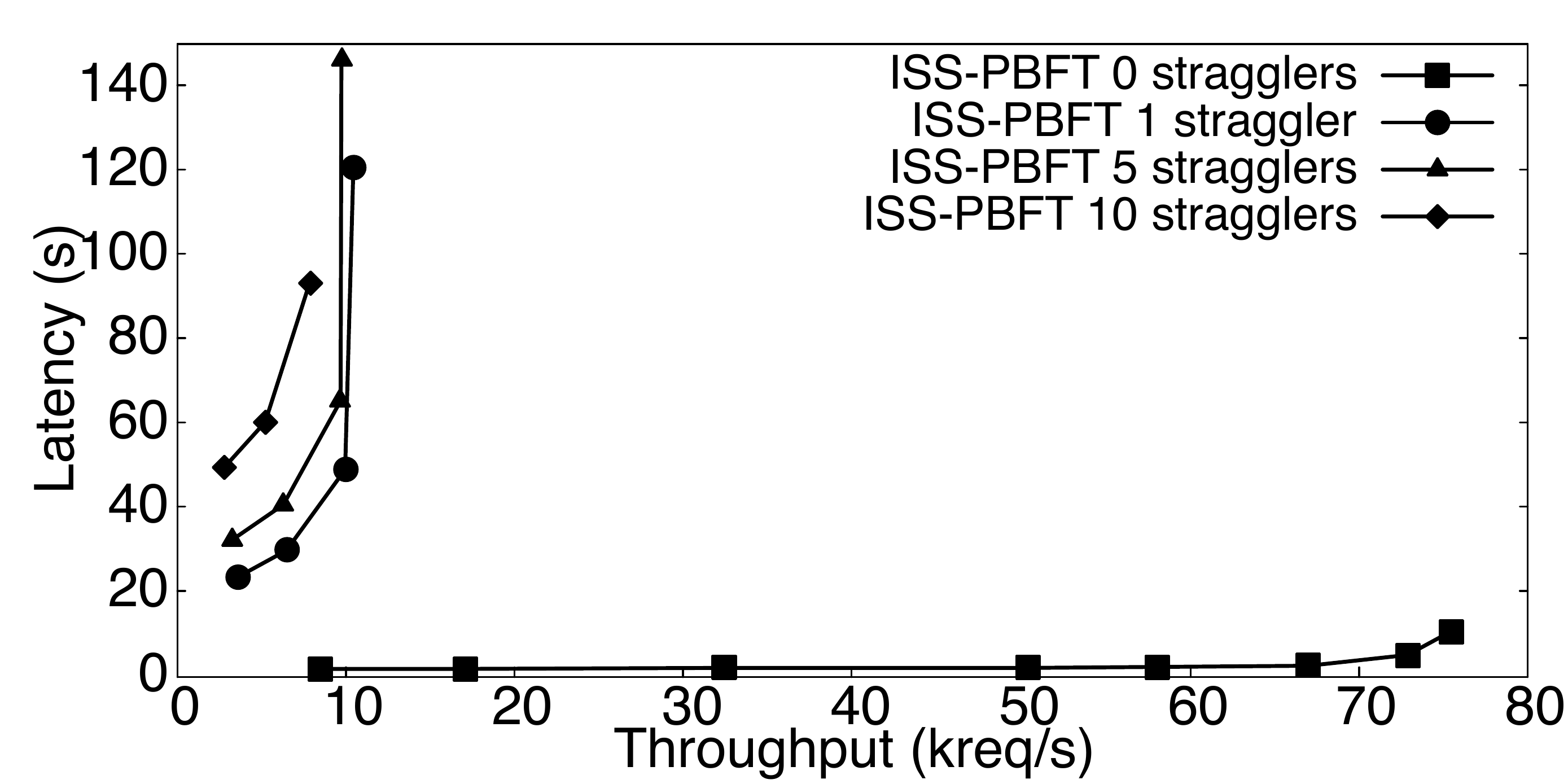}
\caption{\sysname-PBFT latency over throughput for an increasing number of stragglers with Blacklist policy.}
\label{fig:lt-stragglers}
\end{figure}

\Cref{fig:straggler-timeline} shows how throughput is affected over time with a total submission rate of $16.4$kreqs/s.
The performance degradation is due to the ``holes'' in the log temporarily created by the stragglers.
Request delivery progresses as fast as the slowest straggler, hence the spikes in the graph.
When the straggler's batch is finally committed, one more batch per leader can be delivered as well (due to the interleaved batch sequence numbers).
This is inherent to any SMR protocol~\cite{antoniadis2018state} until the straggler is removed from the leaderset.
Straggler resistance in \sysname depends on the underlying \moduleAbbr implementation.
A more sophisticated leader selection policy implementation could dynamically detect and remove stragglers from the leaderset.
\sysname facilitates such dynamic detection by comparing the progress of \moduleAbbr instances,
similarly to RBFT~\cite{RBFT} but without the need for redundant instances.
Alternatively, \moduleAbbr instances could implement Aardvark's~\cite{Aardvark} straggler detection mechanism with reducing timeouts.
This is promising future work.

\begin{figure}[h!]
\centering
\includegraphics[width=0.85\columnwidth]{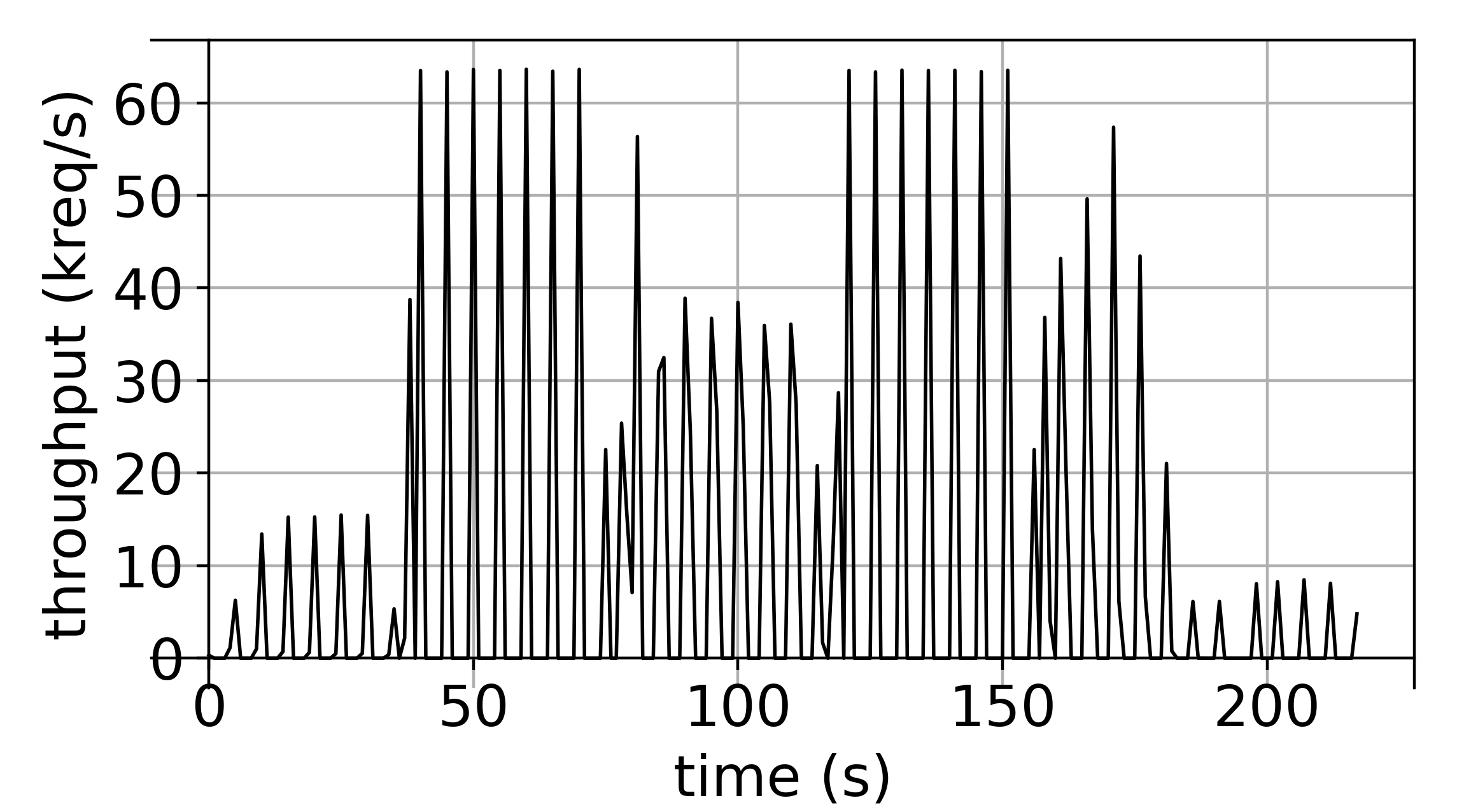}
\caption{\sysname-PBFT throughput average (over 1s intervals) over time with one Byzantine straggler.
Each spike (every $5$ seconds) corresponds to a group of correct leaders' bathces delivered after the straggler's batch.}
\label{fig:straggler-timeline}
\end{figure}

\section{Related Work}
\label{sec:related}

Consensus under Byzantine faults was first made practical by Castro and Liskov~\cite{castro1998practical}\cite{Castro:2002:PBF} who introduced PBFT, a semi-permanent leader-driven protocol.
The quadratic message complexity of PBFT across all replicas triggered vigorous research towards protocols with linear message complexity.
Ramasamy and Cachin~\cite{Parsimonious} replace reliable broadcast in the common case (fault-free execution) with echo broadcast, achieving common case message complexity $O(n)$ per delivered payload.
Echo broadcast is also exploited in \cite{Zyzzyva}\cite{SBFT} to achieve linear common case message complexity.
Only recently, HotStuff~\cite{hotstuff}, introduced a 4th communication round to the 3 message rounds of reliable broadcast, to achieve linear message complexity also for the recovery phase (view-change) of the protocol.
Regardless of the improvement of message complexity, all aforementioned protocols have a single leader at a time, persistent or not, limiting throughput scalability.

Mencius\cite{Mencius} introduced multiple parallel leaders, running instances of Paxos~\cite{paxos}, to achieve throughput scalability and low latency in WAN under crash fault assumptions.
BFT-Mencius~\cite{BFT-Mencius} was the first work to introduce parallel leaders under byzantine faults.
BFT-Mencius introduced the Abortable Timely Announced Broadcast communication primitive to guarantee bounded delay after GST.
However, BFT-Mencius, partitions requests among instances by deterministically assigning clients to replicas, which cannot guarantee load balancing.
Moreover, this opens a surface to duplication performance attacks, since malicious clients and replicas can abuse the suggested denial of service mitigation mechanism.

Guerraoui \textit{et al.}~\cite{700paper} also introduced an abstraction which allows BFT instances to abort.
The paper uses the abstraction to compose sequentially different BFT protocols, which allows a system to choose the optimal protocol according to network conditions.
Our work, on the other hand, composes TOB instances in parallel to achieve throughput scalability.

Mir-BFT~\cite{stathakopoulou2019mir} is the multi-leader protocol which eliminates request duplication \textit{ante} broadcast, effectively preventing duplication attacks.
Later FnF~\cite{avarikioti2020fnf} suggested improved leaderset policies for throughput scalability under performance attacks.
FnF adopts MirBFT's request space partitioning mechanism for duplication prevention.
Similarly, Dandelion\cite{dandelion} leverages the same mechanism to combine Algorand\cite{Algorand} instances.
However, Mir-BFT and FnF multiplex PBFT and SBFT~\cite{SBFT} instances, respectively, leveraging a single replica in the role of epoch primary.
\sysname not only eliminates the need for an epoch primary but also provides a modular framework to multiplex any single leader protocol that can implement \moduleAbbr.

Parallel to this work, several works attempt multiplexing BFT instances to achieve high throughput (Redbelly\cite{crain2018dbft}, RCC\cite{rcc}, Omada\cite{omada}).
However, similarly to BFT-Mencius, clients are assigned to primaries, and, after a timeout, a client can change primary to guarantee liveness, again allowing duplication attacks.

\section{Conclusion}
\label{sec:conclusion}
In this work we introduce \sysname, a general construction for efficiently multiplexing instances of leader-based TOB protocols and increasing their throughput.
\sysname leverages request space partitioning to prevent duplicate requests similarly to MirBFT, but
rotates the partition assignment without the need of a replica to act as a primary, even in the case of faults.
To achieve this, we introduced a Sequenced Broadcast, a novel abstraction that generalizes leader-based TOB protocols and which allows periodically terminating and synchronizing their otherwise independent instances.
Our evaluation shows that our careful engineering in \sysname implementation along with the multi-leader paradigm
indeed results in scalable performance for three single leader protocols (PBFT\cite{Castro:2002:PBF}, HotStuff\cite{hotstuff}, and Raft~\cite{RAFT}), outperforming their original designs by an order of magnitude at scale.

\bibliographystyle{plain}

\end{document}